%% file: bohn_j1.tex
\RequirePackage{xcolor}     
\documentclass[journal,twoside,web]{ieeecolor}
\input{00_Config/IEEE_packages.tex}
\input{00_Config/CBo_packages.tex}
\input{00_Config/IEEE_config.tex}
\input{00_Config/CBo_config.tex}
\input{00_Config/myMacros.tex}

\input{00_Config/myVars.tex}
\input{00_Config/CBo_gloss.tex}

\begin{document}
\acceptednotice
\input{01_Tex/00_title_author_thanks.tex}
\input{01_Tex/00_abstract_CBo.tex}
\input{01_Tex/01_introduction.tex}
\input{01_Tex/02_problem_formulation.tex}
\input{01_Tex/03a_captivity_escape_game.tex}
\input{01_Tex/04_game2quant_and_ctrl.tex}
\input{01_Tex/05_quant_and_ctrl.tex}
\input{01_Tex/06_numerical_example.tex}
\input{01_Tex/07_conclusion.tex}

\input{01_Tex/99_bib.tex}

\end{document}

%% file: 00_Config/IEEE_packages.tex
\usepackage{generic}
\usepackage{amsmath,amssymb,amsfonts}
\usepackage{algorithmic}
\usepackage{graphicx}
\usepackage{algorithm,algorithmic}
\usepackage{textcomp}

%% file: 00_Config/CBo_packages.tex
\usepackage[USenglish]{babel}
\usepackage{csquotes}

\usepackage[
        style=ieee,             
        backend=biber,          
        isbn=false,             
        doi=false,              
        giveninits=true,        
        url=false,              
        date=year,              
        maxbibnames=99,         
        minnames=1,             
        maxnames=5,             
        sorting=none,           
        alldates=year
]{biblatex}

\usepackage[nonumberlist,nomain,nogroupskip,section=section,automake]{glossaries}       

\usepackage{mathtools}  
\usepackage{bm}         
\usepackage{siunitx}	
\usepackage{xkeyval}            
\usepackage[
    color=purple!20,
    textsize=footnotesize,
    textwidth=28pt
    ]{todonotes}		

\usepackage{hyphenat}           
\usepackage{enumerate}
\usepackage{paralist}

\usepackage{caption}
\usepackage{tikz}
\usepackage{pgfplots}
\usetikzlibrary{external}
\pgfplotsset{compat=newest}

\usepgfplotslibrary{fillbetween}
\usetikzlibrary{math} 
\tikzmath{
    \figMaxErrXul = 2.4;
    \figMaxErrYul = -.7*sin(deg(1+(\figMaxErrXul-2))) + .5*(\figMaxErrXul-2) + 1.75;
    \figMaxErrXlr = 2.63; 
    \figMaxErrYlr = sin(deg(\figMaxErrXlr-0.5)) + \figMaxErrXlr - 2.75;
    \figMaxErrMarkx1 = 1;
    \figMaxErrMarky1 = .06*(\figMaxErrMarkx1^2) + .625;
    \figMaxErrMarkx2 = 4;
    \figMaxErrMarky2 = .06*(\figMaxErrMarkx2^2) + .625;
    \figGameDiffEndBetaX = 1.75 - 1/sqrt(2);
    \figGameDiffEndBetaY = 14/10 + 1/sqrt(2);
}

\usepackage{hyperref}
\hypersetup{hidelinks}

\usepackage{caption}
\usepackage{subcaption}

\usepackage{mathrsfs}

\usepackage{eso-pic}

%% file: 00_Config/IEEE_config.tex
\def\BibTeX{{\rm B\kern-.05em{\sc i\kern-.025em b}\kern-.08em
    T\kern-.1667em\lower.7ex\hbox{E}\kern-.125emX}}

%% file: 00_Config/CBo_config.tex
\newtheorem{theorem}{Theorem}
\newtheorem{definition}{Definition}
\newtheorem{proposition}{Proposition}
\newtheorem{lemma}{Lemma}
\newtheorem{corollary}{Corollary}
\newtheorem{remark}{Remark}
\newtheorem{assumption}{Assumption}

\newtheorem{problem}{Problem}
\newtheorem{innerobjective}{Objective}
\newenvironment{objective}[1]
  {\renewcommand\theinnerobjective{#1}\innerobjective}
  {\endinnerobjective}

\addto\extrasUSenglish{%
}

\newglossary[alg]{abb}{asi}{aso}{Abbreviations}
\makeglossaries

\definecolor{myblue}{rgb}{0.00000,0.26667,0.58431}
\definecolor{mygreen}{rgb}{0.65000,0.81000,0.23000}
\definecolor{mygreenfill}{rgb}{0.74902,1.00000,0.00000} 
\definecolor{darkgreen}{RGB}{56, 110, 35}%
\definecolor{darkred}{rgb}{0.38431,0.04314,0.04314}%

\newcommand\acceptedtext{%
  \textcopyright~2026 IEEE. Personal use of this material is permitted.
  Permission from IEEE must be obtained for all other uses, in any current
  or future media, including reprinting/republishing this material for
  advertising or promotional purposes, creating new collective works, for
  resale or redistribution to servers or lists, or reuse of any copyrighted
  component of this work in other works.\\[2pt]
  This is the author's accepted version of the article published in
  IEEE Transactions on Automatic Control. DOI:~10.1109/TAC.2026.3714227%
}

\newcommand\acceptednotice{%
  \AddToShipoutPictureBG*{%
    \begin{tikzpicture}[remember picture,overlay]
      \node[anchor=south, yshift=3mm] at (current page.south)
        {\parbox{\textwidth}{%
          \fontsize{6.75}{7.8}\selectfont
          \acceptedtext}};
    \end{tikzpicture}%
  }%
}

%% file: 00_Config/myMacros.tex
	\newcommand{\Abs}[1][]{\lvert #1 \rvert}        
	\DeclareMathOperator{\signFunc}{sign}				

	\newcommand{\addWithPreComma}[1]{%
  		\if\relax\detokenize{#1}\relax\else,#1\fi
	}

	\newcommand{\combineIndices}[2]{%
	\if\relax\detokenize{#1}\relax
		#2%
	\else
		#1\if\relax\detokenize{#2}\relax\else,#2\fi
	\fi
	}

	\DeclareMathOperator{\diag}{diag}					
	
	\newcommand{\diff}{\mathrm{d}}                              				
	\newcommand{\dPartial}[2]{\dfrac{\partial #1}{\partial #2}}					
	\newcommand{\dFull}[2]{\dfrac{\diff #1}{\diff #2}}							

	\renewcommand{\vec}[1]{\bm{#1}}						
	\renewcommand{\matrix}[1]{\bm{#1}}					

	\newcommand{\Transpose}{\!\top}             			
	\newcommand{\Euclidean}[1][]{\lVert #1 \rVert_{2}}  

	\newcommand{\vertM}{\;\big\vert\;}

	\newcommand{\rBrack}[1][]{( #1 )}

	\newcommand{\cBrack}[1][]{\{ #1 \}}

	\newcommand{\oft}{(\stime)}     			
	\newcommand{\oftZero}{(\stime[\denoteZero])}    

	\newcommand{\identPE}{\textit{PL}}               
	\newcommand{\identPP}{\textit{PH}}               


%% file: 00_Config/myVars.tex
\newcommand{\varBase}[3]{%
	\expandafter\newcommand\csname s#1\endcsname[1][]{%
		#2_{\combineIndices{#3}{##1}}%
	}%
	\expandafter\newcommand\csname s#1dot\endcsname[1][]{%
		\dot{#2}_{\combineIndices{#3}{##1}}%
	}%
	\expandafter\newcommand\csname s#1Gopt\endcsname[1][]{%
		\oSP{#2}_{\combineIndices{#3}{##1}}%
	}%
	\expandafter\newcommand\csname v#1\endcsname[1][]{%
		\vec{#2}_{\combineIndices{#3}{##1}}%
	}%
	\expandafter\newcommand\csname v#1dot\endcsname[1][]{%
		\dot{\vec{#2}}_{\combineIndices{#3}{##1}}%
	}%
	\expandafter\newcommand\csname v#1Gopt\endcsname[1][]{%
		\oSP{\vec{#2}}_{\combineIndices{#3}{##1}}%
	}%
	\expandafter\newcommand\csname v#1T\endcsname[1][]{%
		\vec{#2}_{\combineIndices{#3}{##1}}^{\Transpose}%
	}%
}

\newcommand{\denoteZero}{0}				
\newcommand{\denoteOne}{1}				

\newcommand{\denoteEnd}{\mathrm{end}}	
\newcommand{\denoteGopt}{\star}			
\newcommand{\denoteHf}{\mathrm{h}}		
\newcommand{\denoteLf}{\mathrm{l}}		
\newcommand{\denoteOpt}{\diamond}		
\newcommand{\denotesx}{\mathrm{x1}}

\newcommand{\denotesy}{\mathrm{x2}}

\newcommand{\mapBIP}{\bm{a}}

\newcommand{\vfselXcrit}{\vec{c}}				

\newcommand{\vf}[1][]{\vec{f}_{#1}}

\newcommand{\runJ}{j}
\newcommand{\objectiveFunc}{J}              

\newcommand{\sizeDefinition}{m}

\newcommand{\sizeX}[1][]{n_{#1}}			

\newcommand{\ppTOvPB}{\vec{o}}

\newcommand{\matrxEntry}[1][]{p_{#1}}
\newcommand{\matrx}[1][]{\matrix{P}_{#1}}

\newcommand{\sizeU}[1][]{q_{#1}}			
\newcommand{\matchX}{\matrix{Q}}            

\newcommand{\sizePS}{s}						

\newcommand{\stime}[1][]{t_{#1}}			
\newcommand{\stimeEnd}[1][]{T_{#1}}			

\varBase{u}{u}{}

\newcommand{\sv}[1][]{v_{#1}}				
\newcommand{\valueFunc}{V}                  

\varBase{x}{x}{}
\varBase{xOne}{x}{1}
\varBase{xTwo}{x}{2}

\newcommand{\sizeMParams}[1][]{z_{#1}}		

\newcommand{\ssafetyMargin}{\alpha}			

\newcommand{\sizeScap}{\beta}				

\newcommand{\sStrategyNA}[1][]{\gamma_{#1}}			
\newcommand{\vStrategyNA}[1][]{\vec{\gamma}_{#1}}	

\newcommand{\randomAngle}{\delta}			
\newcommand{\boundary}{\partial}			

\newcommand{\wte}{\zeta}					

\newcommand{\vonScap}[1][]{\vec{\eta}_{#1}}	

\newcommand{\vPB}{\vec{\theta}}				
\newcommand{\pp}{\vartheta}					

\newcommand{\vmp}[1][]{\vec{\iota}_{#1}}	

\newcommand{\vparamBIP}{\bm{\kappa}}

\newcommand{\Scap}{\Lambda}					

\newcommand{\vcondSol}{\vec{\nu}}    

\newcommand{\vSolTrj}[1][]{\vec{\xi}_{#1}}       

\newcommand{\randomVar}{\rho}

\newcommand{\sStrategyFB}[1][]{\varsigma_{#1}}			
\newcommand{\vStrategyFB}[1][]{\vec{\varsigma}_{#1}}	

\newcommand{\stimeBip}{\tau}                   

\newcommand{\relXTrans}{\Phi}				
\newcommand{\so}[1][]{\varphi_{#1}}                   

\newcommand{\vparametrParams}[1][]{\vec{\chi}_{#1}}

\newcommand{\sonBarr}[1][]{\psi_{#1}}					
\newcommand{\vonBarr}[1][]{\vec{\psi}_{#1}}				
\newcommand{\vonBarrdot}[1][]{\dot{\vec{\psi}}_{#1}}	

\newcommand{\syaw}[1][]{\omega_{#1}}

\newcommand{\vSCtrlLaw}{\overline{\vStrategyFB}}

\newcommand{\Steb}{\mathcal{B}}        	

\newcommand{\SCRegion}{\mathcal{C}}

\newcommand{\Sbarr}{\mathcal{K}}                

\newcommand{\Ssemsur}{\mathcal{L}}              

\newcommand{\Smanifold}{\mathcal{M}}			

\newcommand{\Naturals}{\ensuremath{\mathbb{N}}}		

\newcommand{\powerset}{\mathcal{P}}

\newcommand{\Reals}[1][]{\ensuremath{\mathbb{R}^{#1}}}		

\newcommand{\SBIP}{\mathcal{S}}				

\newcommand{\admissU}[1][]{\mathcal{U}_{#1}}	

\newcommand{\Szesc}{\mathcal{V}^-}      			
\newcommand{\Szcap}{\mathcal{V}^+}      			
\newcommand{\Szcapmax}{\hat{\mathcal{V}}^+}      	

\newcommand{\admissX}[1][]{\mathcal{X}_{#1}}    

\newcommand{\reldyn}{\vf^{\vmp}\rBrack[{\vx\oft,\vu[\denoteLf]\oft,\vu[\denoteHf]\oft}]}    
\newcommand{\reldynNE}{\vf^{\vmp}\rBrack[{\vx,\su[\denoteLf],\su[\denoteHf]}]}  

%% file: 00_Config/CBo_gloss.tex
\newacronym{BIP}{BIP}{boundary of the inward-facing part}
\newacronym{CAS}{CAS}{computer algebra system}
\newacronym{fastrack}{FaSTrack}{fast and safe tracking}
\newacronym{MEM}{MEM}{maximum error measure}
\newacronym{WTE}{WTE}{worst-case tracking error}
\newacronym{IP}{IP}{inward-facing part}
\newacronym{ODE}{ODE}{ordinary differential equation}
\newacronym{PDE}{PDE}{partial differential equation}
\newacronym{RTD}{RTD}{Reachability\hyp{}based Trajectory Design}
\newacronym{SOS}{SOS}{sum-of-squares}
\newacronym{TEB}{TEB}{tracking error bound}

%% file: 01_Tex/00_title_author_thanks.tex
\newcommand{\myTitle}[1][]{Captivity-Escape Games as a Means for Safety in Online Motion Generation}

\title{\myTitle}
\author{Christopher Bohn, Manuel Hess, and S\"{o}ren Hohmann, \IEEEmembership{Senior Member, IEEE}
\thanks{The authors are with the Institute of Control Systems, Karlsruhe Institute of Technology, 76131 Karlsruhe, Germany. Corresponding author is Christopher Bohn (e-mail: christopher.bohn@kit.edu).}
}

\maketitle

%% file: 01_Tex/00_abstract_CBo.tex
\begin{abstract}
    This paper addresses conservatism, limited numerical accuracy, and high computational effort in existing methods ensuring safety by design in online model-based motion generation. The presented method employs a novel captivity-escape zero-sum differential game to adapt the planning model’s performance so that resulting reference trajectories are trackable within a prescribed safety margin by a jointly synthesized safety controller. A numerical example demonstrates orders-of-magnitude faster computation and improved numerical accuracy compared to the state of the art.
\end{abstract}


\glsresetall

%% file: 01_Tex/01_introduction.tex
\section{Introduction}\label{sec:02_introduction}
Motion generation comprises two tasks: motion planning and tracking of reference trajectories \autocite{herbert.2017}, \autocite{chen.2021}, \autocite{singh.2020}, \autocite{kousik.2020}. Both tasks are crucial for the functional safety of autonomous systems, so formally ensuring the safety of motion generation is essential \autocite{herbert.2017}, \autocite{chen.2021}, \autocite{singh.2020}, \autocite{kousik.2020}, \autocite{schurmann.2022}, \autocite{fisac.2019}. Formally ensuring safety requires accurate consideration of real-world system dynamics through the use of a high-fidelity model \autocite{herbert.2017}, \autocite{chen.2021}, \autocite{singh.2020}, \autocite{kousik.2020}.

Let $ \Naturals = \{0,1,2,\dots\} $ denote the set of natural numbers. For $ \sizeDefinition\in\Naturals $, let $ \Reals[\sizeDefinition] $ denote the $ \sizeDefinition $-dimensional real vector space.

\begin{definition}\label{def:model_hf}
    A high-fidelity model of a dynamic system is given by a vector $ \vx[\denoteHf]\oft \in \admissX[\denoteHf]\subseteq\Reals[\sizeX] $ containing $ \sizeX \in \Naturals $ states, $ \vx[\denoteHf]: \Reals \rightarrow \admissX[\denoteHf] $, a vector $ \vu[\denoteHf]\oft \in \admissU[\denoteHf] \subseteq \Reals[{\sizeU[\denoteHf]}] $ containing $ \sizeU[\denoteHf] \in \Naturals $ inputs, $ \vu[\denoteHf]: \Reals \rightarrow \admissU[\denoteHf] $, a vector $ \vmp[\denoteHf] \in \Reals[{\sizeMParams[\denoteHf]}] $ containing $\sizeMParams[\denoteHf] \in \Naturals $ parameters, and an \gls*{ODE}
    \begin{equation}\label{eq:model_hf}
        \vxdot[\denoteHf]\oft = \vf[\denoteHf]^{\vmp[\denoteHf]}\rBrack[{\vx[\denoteHf]\oft, \vu[\denoteHf]\oft }].
    \end{equation}
    The set $ \admissU[\denoteHf] $ is compact, $ \vf[\denoteHf]^{\vmp[\denoteHf]}: \admissX[\denoteHf] \times \admissU[\denoteHf] \rightarrow \admissX[\denoteHf] $ is Lipschitz continuous in $ \vx[\denoteHf]\oft $ for a fixed $ \vu[\denoteHf]\oft $, and continuously differentiable in $ \vx[\denoteHf]\oft $. Thus, given any measurable function $ \vu[\denoteHf] $ (see \autocite{mitchell.2005}), a unique trajectory exists that solves \eqref{eq:model_hf} \autocite{coddington.1984}.
\end{definition}
\begin{assumption}\label{assump:hf_real_system}
    A high-fidelity model of a dynamic system accurately represents the related real-world system dynamics.
\end{assumption}

Advanced model-based control design methods facilitate the use of a high-fidelity model \autocite{herbert.2017}, \autocite{chen.2021}, \autocite{singh.2020}, \autocite{kousik.2020}, \autocite{kohler.2021}. Therefore, consider a high-fidelity model to be used for tracking.

Conversely, motion planning requires a substantial prediction horizon and fast computation, rendering the use of a high-fidelity model impractical due to computational limitations. Consequently, motion planning commonly uses a low-fidelity model \autocite{herbert.2017}, \autocite{chen.2021}, \autocite{singh.2020}, \autocite{kousik.2020}.
\begin{definition}\label{def:model_lf}
    A low-fidelity model of a dynamic system is given by a vector $ \vx[\denoteLf]\oft \in \admissX[\denoteLf] \subseteq \admissX[\denoteHf] $, with $ \dim \admissX[\denoteLf] = \sizeX[\denoteLf] $, containing $ \sizeX[\denoteLf] \in \Naturals $ states, $ \vx[\denoteLf]: \Reals \rightarrow \admissX[\denoteLf] $, a vector $ \vu[\denoteLf]\oft \in \admissU[\denoteLf] \subseteq \Reals[{\sizeU[\denoteLf]}] $ containing $ \sizeU[\denoteLf] \in \Naturals $ inputs, $ \vu[\denoteLf]: \Reals \rightarrow \admissU[\denoteLf] $, a vector $ \vmp[\denoteLf] \in \Reals[{\sizeMParams[\denoteLf]}] $ containing $\sizeMParams[\denoteLf] \in \Naturals $ parameters, and an \gls*{ODE}
    \begin{equation}\label{eq:model_lf}
        \vxdot[\denoteLf]\oft = \vf[\denoteLf]^{\vmp[\denoteLf]}\rBrack[{ \vx[\denoteLf]\oft, \vu[\denoteLf]\oft }].
    \end{equation}
    \noindent The properties of $ \admissU[\denoteLf] $ and $ \vf[\denoteLf]^{\vmp[\denoteLf]} $ match those of $ \admissU[\denoteHf] $ and $ \vf[\denoteHf]^{\vmp[\denoteHf]} $.
\end{definition}
\begin{figure}[!t]
    \centering
    \input{02_Img/figure_start_goal_unsafe.tex}
    \caption{
    The figure illustrates the task of planning a reference trajectory from $ \vx[\denoteLf](\stime[\denoteZero]) $ to $ \vx[\denoteLf](\stime[\denoteOne]) $, indicated by the blue arrow. The gray regions indicate unsafe states (e.g., obstacles), the black lines indicate safety-critical constraints, and the green regions indicate the safety margin $ \ssafetyMargin $.
    }
    \label{fig:start_goal_unsafe}
\end{figure}

Crucially, references planned with a low-fidelity model are not necessarily dynamically feasible for the high-fidelity model used for tracking. This potentially results in unavoidable tracking errors that may cause violations of safety-critical constraints, thus posing a safety risk in motion generation.

To address this safety risk, the frameworks in \autocite{herbert.2017}, \autocite{chen.2021}, \autocite{singh.2020}, \autocite{kousik.2020} augment safety-critical constraints in motion planning by a safety margin $ \ssafetyMargin \in \Reals[+] $ that ensures preventing constraint violations under the \gls*{WTE} (see \autoref{fig:start_goal_unsafe}). In addition, a safety controller is used to ensure the tracking error does not exceed $ \ssafetyMargin $. The methods employed in these frameworks determine both the safety margin $ \ssafetyMargin $ and controller based on a specific pair of low-fidelity and high-fidelity models. Note that $ \admissX[\denoteHf] $, $ \admissU[\denoteHf] $, and $ \vmp[\denoteHf] $ can be determined via system identification, whereas $ \admissX[\denoteLf] $, $ \admissU[\denoteLf] $, and $ \vmp[\denoteLf] $ must be selected heuristically to suit the planning task. 
We denote the objective of determining $ \ssafetyMargin $ for a specific pair of models as \ref{O.0}.
\begin{objective}{O.0}\label{O.0}
    Compute a safety margin $ \ssafetyMargin $ that ensures safe motion generation with given $ \vf[\denoteLf]^{\vmp[\denoteLf]} $, $ \admissX[\denoteLf] $, $ \admissU[\denoteLf] $, $ \vmp[\denoteLf] $, $ \vf[\denoteHf]^{\vmp[\denoteHf]} $, $ \admissX[\denoteHf] $, $ \admissU[\denoteHf] $, $ \vmp[\denoteHf] $.
\end{objective}

However, addressing \ref{O.0} with the methods in \autocite{herbert.2017}, \autocite{chen.2021}, \autocite{singh.2020} is computationally intensive, and the numerical accuracy of the safety margin $ \ssafetyMargin $ resulting from \autocite{herbert.2017}, \autocite{chen.2021} is not ensured \autocite{kousik.2020}. Moreover, the heuristically selected $ \admissX[\denoteLf] $, $ \admissU[\denoteLf] $, and $ \vmp[\denoteLf] $, and the resulting safety margin $ \ssafetyMargin $ are likely ill-suited to a given planning task, leading to overly conservative references.\footnote{Restrictive $ \admissX[\denoteLf] $, $ \admissU[\denoteLf] $, $ \vmp[\denoteLf] $ result in low-performance references; nonrestrictive $ \admissX[\denoteLf] $, $ \admissU[\denoteLf] $, $ \vmp[\denoteLf] $ require a large safety margin, inducing conservatism by shrinking the planning space and potentially rendering the planning task infeasible.
} 

To mitigate the existing framework's conservatism, $ \admissX[\denoteLf] $, $ \admissU[\denoteLf] $, and $ \vmp[\denoteLf] $ must be adapted to a safety margin $ \ssafetyMargin $ that suits the specific planning task. However, adapting $ \admissX[\denoteLf] $, $ \admissU[\denoteLf] $, $ \vmp[\denoteLf] $ to a given $ \ssafetyMargin $ using the methods in \autocite{herbert.2017}, \autocite{chen.2021}, \autocite{singh.2020} entails iteratively tuning $ \admissX[\denoteLf] $, $ \admissU[\denoteLf] $, $ \vmp[\denoteLf] $ and recomputing $ \ssafetyMargin $, which is computationally intractable.

In this paper, we address these limitations by presenting a computationally efficient and numerically accurate method that adapts $ \admissX[\denoteLf] $, $ \admissU[\denoteLf] $, and $ \vmp[\denoteLf] $ to a given safety margin $ \ssafetyMargin $, thereby focusing on the following objectives:
\begin{objective}{O.1}\label{O.1}
    Directly compute $ \admissX[\denoteLf] $, $ \admissU[\denoteLf] $, and $ \vmp[\denoteLf] $ that ensure safe motion generation with given $ \ssafetyMargin $, $ \vf[\denoteLf]^{\vmp[\denoteLf]} \! $, $ \vf[\denoteHf]^{\vmp[\denoteHf]} \! $, $ \admissX[\denoteHf] $, $ \admissU[\denoteHf] $, and $ \vmp[\denoteHf] $.
\end{objective}
\begin{objective}{O.2}\label{O.2}
    Design a safety controller that ensures safe tracking within a given safety margin $ \ssafetyMargin $, given $ \vf[\denoteHf]^{\vmp[\denoteHf]} $, $ \admissX[\denoteHf] $, $ \admissU[\denoteHf] $, $ \vmp[\denoteHf] $, and $ \vf[\denoteLf]^{\vmp[\denoteLf]} $, with $ \admissX[\denoteLf] $, $ \admissU[\denoteLf] $, and $ \vmp[\denoteLf] $ as determined per \ref{O.1}.
\end{objective}
\subsection{Related Work}\label{sec:related_work}
In \autocite{herbert.2017}, \autocite{chen.2021}, the authors present the state-of-the-art method for addressing \ref{O.0} and the state-of-the-art framework, \gls*{fastrack}, which uses a safety margin to ensure safe motion generation. The method presented in \autocite{herbert.2017}, \autocite{chen.2021} uses a pursuit-evasion differential game to determine a safety margin and the associated safety controller. The related computations are intensive, as the utilized pursuit-evasion game is solved numerically using Hamilton-Jacobi reachability analysis \autocite{mitchell.2005}.\footnote{Determining a safety margin with decent accuracy takes days, even if both models are simple and low-dimensional.} As a consequence of the numerical methods involved, the numerical accuracy of the resulting safety margin cannot be ensured \autocite{kousik.2020}. This poses a safety risk in \gls*{fastrack}. Furthermore, the \gls*{fastrack} framework uses a fixed safety margin and fixed $ \admissX[\denoteLf] $, $ \admissU[\denoteLf] $, and $ \vmp[\denoteLf] $ across planning environments, which likely results in planning overly conservative references \autocite{fridovich-keil.2018}.

In \autocite{fridovich-keil.2018}, the authors build upon \gls*{fastrack} and reduce conservatism by using a set of motion planners that feature different safety margins. A meta-planning algorithm is used to online select a planner that suits the planning environment. However, \autocite{fridovich-keil.2018} remains conservative, as the set of available planners is limited due to \gls*{fastrack}'s computational intensity. In addition, \autocite{fridovich-keil.2018} inherits \gls*{fastrack}'s limited numerical accuracy.

In \autocite{sahraeekhanghah.2021}, a method is presented that builds upon \gls*{fastrack} and reduces conservatism by constraining motion planning to motion primitives. However, this restricts resulting references to compositions of a limited set of primitives.

A method for addressing \ref{O.0} through sum-of-squares optimization is presented in \autocite{singh.2020}. However, compared to \gls*{fastrack}, the method in \autocite{singh.2020} results in a more conservative safety margin.

The authors in \autocite{kousik.2020} present a motion generation framework that utilizes a safety margin but do not focus on \ref{O.0}.

The objective \ref{O.1} remains unaddressed in the literature. All related works focus on \ref{O.0}, necessitating heuristic selection of $\admissX[\denoteLf]$, $\admissU[\denoteLf]$, and $\vmp[\denoteLf]$. In addition, all existing methods are computationally intensive, and the methods based on the state of the art (\gls*{fastrack}) do not ensure numerical accuracy.
\subsection{Contribution}
The contribution of this work is a method that addresses \ref{O.1}, \ref{O.2}, and also \ref{O.0}. The presented method is based on a captivity-escape game, a novel zero-sum differential game with two players introduced in this work. The results obtained using the presented method are numerically accurate\footnote{Depending on the employed models, computations are either performed analytically or using numerical integration. Approaches exist for determining the errors resulting from numerical integration to ensure safety \autocite[Chap. 9]{quarteroni.2010}.} and the computation time of the presented method is significantly faster\footnote{We consider the computation time of the presented method to be within the timescale that allows for online adaptation of $ \admissX[\denoteLf] $, $ \admissU[\denoteLf] $, and $ \vmp[\denoteLf] $.} than the computation time of the state-of-the-art method \gls*{fastrack}. Consequently, our method complements the frameworks presented in \autocite{herbert.2017}, \autocite{chen.2021}, \autocite{singh.2020}, \autocite{kousik.2020} by reducing computation time, ensuring accuracy, and mitigating conservatism. We demonstrate the presented method using a numerical example and compare it to the state-of-the-art method in \gls*{fastrack}.
\subsection{Outline}
In \autoref{sec:02_prob_def}, we formulate the problems related to \ref{O.1} and \ref{O.2}. In \autoref{sec:03}, we introduce captivity-escape differential games, which we leverage in \autoref{sec:04_relation} to address the problems related to \ref{O.1} and \ref{O.2}. Solving captivity-escape games is addressed in \autoref{sec:05_solving}, and a numerical example in \autoref{sec:06_numerical_example} demonstrates the presented method.

%% file: 02_Img/figure_start_goal_unsafe.tex
\begin{tikzpicture}
    \pgfmathdeclarefunction{usBfunc}{1}{%
        \pgfmathparse{sin(deg(#1-0.5)) + #1 - 2.75}%
    }
    \pgfmathdeclarefunction{dusBfunc}{1}{%
    \pgfmathparse{cos(deg(#1-0.5)) + 1}%
    }
    \pgfmathdeclarefunction{usTfunc}{1}{%
        \pgfmathparse{-.7*sin(deg(1+(#1 - 2.217))) + .5*(#1 - 2.217) + 2}%
    }
    \pgfmathdeclarefunction{dusTfunc}{1}{%
    \pgfmathparse{-.7*cos(deg(1+(#1 - 2.217))) + .5}%
    }
    \pgfmathdeclarefunction{pathfunc}{1}{%
    \pgfmathparse{.05 * (#1 + 0.613118701612116)^2 + 0.598119067714960 + 0.002}%
    }
    \pgfmathdeclarefunction{dpathfunc}{1}{%
    \pgfmathparse{.1 * (#1 + 0.613118701612116)}%
    }

    \def\pathStartX{1}
    \def\pathStopX{4}
    \pgfmathsetmacro{\pathStartY}{pathfunc(\pathStartX)}
    \pgfmathsetmacro{\pathStopY}{pathfunc(\pathStopX)}

    \def\closestBX{2.7924}    
    \def\closestTX{2.5564}    
    \pgfmathsetmacro{\closestBY}{usBfunc(\closestBX)}
    \pgfmathsetmacro{\closestTY}{usTfunc(\closestTX)}

    \def\rTEBs{.2}
    \def\rTEBa{.25}
    \def\rTEBl{.45}

    \pgfmathsetmacro{\closestBXA}{\closestBX - (\rTEBa * dusBfunc(\closestBX)) / sqrt(1 + dusBfunc(\closestBX)^2)}
    \pgfmathsetmacro{\closestBYA}{usBfunc(\closestBX) + \rTEBa / sqrt(1 + dusBfunc(\closestBX)^2)}

    \pgfmathsetmacro{\closestTXA}{\closestTX + (\rTEBa * dusTfunc(\closestTX)) / sqrt(1 + dusTfunc(\closestTX)^2)}
    \pgfmathsetmacro{\closestTYA}{usTfunc(\closestTX) - \rTEBa / sqrt(1 + dusTfunc(\closestTX)^2)}



    \begin{axis}[
        xmin=0, xmax=5,
        ymin=.32, ymax=1.95,
        axis equal image,
        scale only axis,
        width=0.8\columnwidth,
        hide axis,
        ]

        \addplot [domain=0:5, samples=100, thick, color=black!80, name path=bottomright]
        {usBfunc(x)};
        \addplot [domain=0:5, samples=100, thick, color=white, name path=lower]
        {-10};
        \addplot[black!20, opacity=0.4] fill between [of=bottomright and lower];
    
        \addplot [domain=0:5, samples=100, thick, color=black!80, name path=topleft]
        {usTfunc(x)};
        \addplot [domain=0:5, samples=100, thick, color=white, name path=upper]
        {10};
        \addplot[black!20, opacity=0.4] fill between [of=topleft and upper];

        \addplot [-stealth, domain=\pathStartX+.1:\pathStopX-.1, samples=100, very thick, color=myblue]
        {pathfunc(x)};

        \addplot [domain=-1:6, samples=100, mark=none, thick, mygreen, name path = TEBaB] (
            {x + (\rTEBa * dusTfunc(x)) / sqrt(1 + dusTfunc(x)^2)},
            {usTfunc(x) - \rTEBa / sqrt(1 + dusTfunc(x)^2)}
        );
        \addplot[mygreenfill!20] fill between [of=TEBaB and topleft];

        \addplot [domain=-1:6, samples=100, mark=none, thick, mygreen, name path = TEBaB] (
            {x - (\rTEBa * dusBfunc(x)) / sqrt(1 + dusBfunc(x)^2)},
            {usBfunc(x) + \rTEBa / sqrt(1 + dusBfunc(x)^2)}
        );
        \addplot[mygreenfill!20] fill between [of=TEBaB and bottomright];


        \addplot [style={solid, fill=red}, only marks]
        coordinates {(\pathStartX,\pathStartY) (\pathStopX,\pathStopY)};

        \node[color=black, anchor=north, font=\footnotesize] at (axis cs: 1.9,1.875) {unsafe states};
        \node[color=black, anchor=south, font=\footnotesize] at (axis cs: 3.75,.44) {unsafe states};
        \node[color=black, font=\footnotesize] at (axis cs: \pathStartX,\pathStartY-.21) {$ \vx[\denoteLf](\stime[\denoteZero]) $};
        \node[color=black, font=\footnotesize] at (axis cs: \pathStopX,\pathStopY-.21) {$ \vx[\denoteLf](\stime[\denoteOne]) $};

        \addplot [stealth-stealth, color=black!80]
        coordinates {
            (\closestBX, \closestBY )
            (\closestBXA,\closestBYA)
        };
        \node[color=black, font=\footnotesize, anchor = west] at (axis cs: \closestBX-.05,\closestBY+.15) {$ \ssafetyMargin $};

        \addplot [stealth-stealth,color=black!80]
        coordinates {
            (\closestTX, \closestTY )
            (\closestTXA,\closestTYA)
        };
        \node[color=black, font=\footnotesize, anchor = west] at (axis cs: \closestTX+.03,\closestTY-.07) {$ \ssafetyMargin $};


        \draw[gray] (rel axis cs:0,0) rectangle (rel axis cs:1,1);
\end{axis}

\end{tikzpicture}

%% file: 01_Tex/02_problem_formulation.tex
\section{Problem Formulation}\label{sec:02_prob_def}
In this section, we formulate the problems related to addressing \ref{O.1} and \ref{O.2}. To this end, we formulate the relative system between a high-fidelity model and a low-fidelity model and define the \gls*{WTE} along with associated quantities.
\subsection{Formulation of the Relative System}\label{subsec:rel_dyn}
The tracking error between a high-fidelity model used for tracking and a low-fidelity model used for planning is given by the relative state vector $ \vx\oft \in \admissX \subseteq \Reals[\sizeX] $, $ \vx: \Reals \rightarrow \admissX $ with
\begin{equation}\label{eq:rel_sys}
    \vx\oft \coloneq \relXTrans\rBrack[{\vx[\denoteLf]\oft,\vx[\denoteHf]\oft}]\rBrack[{\matchX \vx[\denoteLf]\oft - \vx[\denoteHf]\oft}].
\end{equation}
The matrix $ \matchX $ is an embedding matrix from $ \admissX[\denoteLf] $ to $ \admissX[\denoteHf] $, and $ \relXTrans\rBrack[{\vx[\denoteLf]\oft,\vx[\denoteHf]\oft}] $ is a Lipschitz continuous map that must be determined such that the relative system dynamics
\begin{equation}\label{eq:dyn_rel}
    \vxdot\oft = \vf^{\vmp}\rBrack[{\vx\oft, \vu[\denoteLf]\oft, \vu[\denoteHf]\oft}]
\end{equation}
are given by a function $ \vf^{\vmp}: \admissX \times \admissU[\denoteLf] \times \admissU[\denoteHf] \rightarrow \admissX $ that is Lipschitz continuous in $ \vx\oft $ for fixed $ \vu[\denoteLf]\oft $ and $ \vu[\denoteHf]\oft $, and continuously differentiable in $ \vx\oft $, with $ \vmp = \begin{bmatrix} \vmp[\denoteLf] & \vmp[\denoteHf] \end{bmatrix} \in \Reals[ {\sizeMParams[\denoteLf] + \sizeMParams[\denoteHf]} ] $.\footnote{If mobile robots are considered, the map $ \relXTrans\rBrack[{\vx[\denoteLf]\oft,\vx[\denoteHf]\oft}] $ is often given by the identity map or a rotation map for various combinations of high-fidelity and low-fidelity models. A discussion on $ \relXTrans\rBrack[{\vx[\denoteLf]\oft,\vx[\denoteHf]\oft}] $ is given in \autocite{singh.2020}.}

Thus, given an initial state $ \vx\oftZero $ and any measurable input functions $ \vu[\denoteLf] $ and $ \vu[\denoteHf] $, a unique solution trajectory of \eqref{eq:dyn_rel} exists
\begin{equation}\label{eq:trajectory}
    \vSolTrj\rBrack[{\cdot;\stime[\denoteZero],\vx\oftZero,\vu[\denoteLf],\vu[\denoteHf]}]:[\stime[\denoteZero],\stimeEnd]\rightarrow\admissX,
\end{equation}
which satisfies $ \vSolTrj\rBrack[{\stime[\denoteZero];\stime[\denoteZero],\vx\oftZero,\vu[\denoteLf],\vu[\denoteHf]}] = \vx\oftZero $ and
\begin{equation}
    \begin{aligned}
        \dot{\vSolTrj} & \rBrack[{\stime;\stime[\denoteZero],\vx\oftZero,\vu[\denoteLf],\vu[\denoteHf]}] = \\ 
        & \vf^{\vmp}\rBrack[{\vSolTrj\rBrack[{\stime;\stime[\denoteZero],\vx\oftZero,\vu[\denoteLf],\vu[\denoteHf]}],\vu[\denoteLf]\oft,\vu[\denoteHf]\oft}]
    \end{aligned}
\end{equation}
almost everywhere on $ \stime \in [\stime[\denoteZero],\stimeEnd] $ \autocite{mitchell.2005}, \autocite{evans.1984}.
\subsection{Problems Related to Addressing \ref{O.1} and \ref{O.2}}
As addressing \ref{O.1} relates to determining the sets $ \admissX[\denoteLf] $, $ \admissU[\denoteLf] $, and the vector $ \vmp[\denoteLf] $, we assume that $ \admissX[\denoteLf] $, $ \admissU[\denoteLf] $, and $ \vmp[\denoteLf] $ can be determined by means of a single parameter.
\begin{assumption}\label{assump:vPB}
    The sets $ \admissX[\denoteLf] $ and $ \admissU[\denoteLf] $ can be parametrized such that each set is fully defined by a vector $ \vparametrParams[{\admissX[\denoteLf]}] $ and $ \vparametrParams[{\admissU[\denoteLf]}] $, respectively, containing the related parameters.
\end{assumption}
\begin{definition}\label{def:pp}
    Let Assumption~\ref{assump:vPB} hold. Consider $ \vPB \in \Reals[\sizePS] $, with $ \vPB = \begin{bmatrix} \vparametrParams[{\admissX[\denoteLf]}] & \vparametrParams[{\admissU[\denoteLf]}] & \vmp[\denoteLf] \end{bmatrix} $, which fully defines $ \admissX[\denoteLf] $, $ \admissU[\denoteLf] $, and $ \vmp[\denoteLf] $. The mapping $ \ppTOvPB\rBrack[\pp] = \vPB $, $ \ppTOvPB: \Reals[+]\rightarrow \Reals[\sizePS] $ determines $ \admissX[\denoteLf] $, $ \admissU[\denoteLf] $, and $ \vmp[\denoteLf] $ by means of the so-called planning performance $ \pp \in \Reals[+] $.
\end{definition}
\begin{assumption}
    The mapping $ \ppTOvPB\rBrack[\pp] = \vPB $ relates larger values of $ \pp $ to improved performance of the low-fidelity model. 
\end{assumption}

In addition, addressing \ref{O.1} and \ref{O.2} requires a formal relation between $ \pp $, a safety margin $ \ssafetyMargin $, and a safety controller.
\begin{definition}\label{def:fb_strategy}
    The mapping $ \vu[\denoteHf] = \vStrategyFB[\denoteHf]\rBrack[\vx] $, with $ \vStrategyFB[\denoteHf]: (\Reals \rightarrow \admissX) \rightarrow (\Reals \rightarrow \admissU[\denoteHf]) $, denotes a state-feedback strategy.
\end{definition}
\begin{definition}\label{def:na_strategy}
    The mapping $ \vu[\denoteLf] = \vStrategyNA[\denoteLf]\rBrack[{\vx,\vu[\denoteHf]}] $, with $ \vStrategyNA[\denoteLf]: (\Reals \rightarrow \admissX) \times (\Reals \rightarrow \admissU[\denoteHf]) \rightarrow (\Reals \rightarrow \admissU[\denoteLf]) $, denotes a nonanticipative strategy as formulated in \autocite{mitchell.2005}.
\end{definition}

For brevity, we use the notations $ \vStrategyFB[\denoteHf]\rBrack[\vx] $ and $ \vStrategyNA[\denoteLf]\rBrack[{\vx,\vu[\denoteHf]}] $ to refer to a strategy, and the notations $ \vStrategyFB[\denoteHf]\rBrack[\vx\oft] $ and $ \vStrategyNA[\denoteLf](\vx\oft,$ $\vu[\denoteHf]\oft) $ to refer to an input vector resulting from a strategy.
\begin{definition}\label{def:crit_state_vec}
    The critical relative state vector $ \vfselXcrit\rBrack[\vx\oft] \in \admissX $ captures relative states for which safety-critical constraints exist. It is given by the linear mapping $ \vfselXcrit\rBrack[\vx\oft] = \matrx\vx\oft $, $\vfselXcrit: \admissX \rightarrow \admissX $, with $ \matrx = \diag\rBrack[{\matrxEntry[\runJ]}] $, $ j=1\dots\sizeX $, and $ \matrxEntry[\runJ] = 1 $, if safety-critical constraints exist for $ \sx[\runJ] $, and $ \matrxEntry[\runJ] = 0 $, otherwise.
\end{definition}
\begin{definition}\label{def:teb}
    A \gls*{TEB} is a set $ \Steb \subseteq \admissX $ with a nonempty and compact $ \vfselXcrit\rBrack[\Steb] \subseteq \admissX $, and for which a state-feedback strategy $ \vu[\denoteHf] = \vStrategyFB[\denoteHf]\rBrack[\vx] $ exists such that for any nonanticipative strategy $ \vu[\denoteLf] = \vStrategyNA[\denoteLf]\rBrack[{\vx,\vu[\denoteHf]}] $ and any initial $ \vx\oftZero \in \Steb \Rightarrow \forall \stime \geq \stime[\denoteZero]: \vSolTrj(\stime; \stime[\denoteZero],\vx\oftZero,\vStrategyNA[\denoteLf]\rBrack[{\vx,\vu[\denoteHf]}],\vStrategyFB[\denoteHf]\rBrack[{\vx}]) \in\Steb $.
\end{definition}

The boundary of $ \Steb $ is denoted as $ \boundary\Steb $.
\glsreset{WTE}
\begin{definition}\label{def:wte}
    The \gls*{WTE} $ \wte $ is given by $ \wte \coloneq \max_{\vx\oft\in\Steb} \Euclidean[{\vfselXcrit\rBrack[\vx\oft]}] $.
\end{definition}

\begin{remark}\label{rem:smargin}
    To prevent safety-critical constraint violations under the \gls*{WTE} $ \wte $, a safety margin $ \ssafetyMargin $ must fulfill $ \wte \leq \ssafetyMargin $.
\end{remark}

With Remark~\ref{rem:smargin}, a \gls*{TEB} $ \Steb $ and the related \gls*{WTE} $ \wte $ establish the required formal relation between the planning performance $ \pp $ (i.e., $ \admissX[\denoteLf] $, $ \admissU[\denoteLf] $, and $ \vmp[\denoteLf] $, see Definition~\ref{def:pp}), a safety margin $ \ssafetyMargin $, and a safety controller. To minimize conservatism, the \gls*{TEB} $ \Steb $ must be constituted by means of the state-feedback strategy that is optimal w.r.t. retaining solution trajectories of \eqref{eq:dyn_rel} in $ \Steb $. This strategy is denoted as $ \vu[\denoteHf] = \vStrategyFB[\denoteHf]^\bullet\rBrack[\vx] $. Moreover, the strategy that is optimal w.r.t. forcing solution trajectories of \eqref{eq:dyn_rel} to leave the \gls*{TEB} $ \Steb $, denoted as $ \vu[\denoteLf] = \vStrategyNA[\denoteLf]^\bullet\rBrack[{\vx,\vu[\denoteHf]}] $, must be used to ensure the \gls*{TEB} $ \Steb $ holds for any $ \vu[\denoteLf] = \vStrategyNA[\denoteLf]\rBrack[{\vx,\vu[\denoteHf]}] $.
\begin{problem}\label{problem_1}
    Determine the optimal strategies $ \vu[\denoteHf] = \vStrategyFB[\denoteHf]^\bullet\rBrack[\vx] $ and $ \vu[\denoteLf] = \vStrategyNA[\denoteLf]^\bullet\rBrack[{\vx,\vu[\denoteHf]}] $ for constituting a \gls*{TEB} $ \Steb $.
\end{problem}

To maintain focus on minimizing conservatism, the planning performance $ \pp $ must be adapted such that the safety margin $ \ssafetyMargin $ is fully utilized by the related \gls*{WTE} $ \wte $, requiring $ \wte = \ssafetyMargin $.
\begin{problem}\label{problem_2}
    Determine the planning performance $ \pp $ so that the \gls*{TEB} constituted by means of $ \vu[\denoteHf] = \vStrategyFB[\denoteHf]^\bullet\rBrack[\vx] $ and $ \vu[\denoteLf] = \vStrategyNA[\denoteLf]^\bullet\rBrack[{\vx,\vu[\denoteHf]}] $ results in a \gls*{WTE} $ \wte $ that fulfills $ \wte = \ssafetyMargin $.
\end{problem}

%% file: 01_Tex/03a_captivity_escape_game.tex
\section{Captivity-Escape Differential Game}\label{sec:03}
In this section, we formulate a novel differential game that is tailored to efficiently address Problem~\ref{problem_1} and Problem~\ref{problem_2}.
\subsection{Formulation of a Captivity-Escape Game}\label{sec:03_game}
\begin{figure}[!t]
    \centering
    \input{02_Img/figure_game_difference.tex}
    \caption{A valid initial state $ \vx\oftZero $ for both games is depicted by the position of \identPE\ (respectively, $ \tilde{\text{\identPE}} $). Each game terminates when \identPE\ (respectively, $ \tilde{\text{\identPE}} $) exits the green region into the gray region.}
    \label{fig:two_games}
\end{figure}
A captivity-escape game is a deterministic zero-sum differential game with variable terminal time $ \stimeEnd[\denoteEnd] $. In the game, there are two players that we refer to as \identPE\ and \identPP. \identPE\ employs the dynamics of the low-fidelity model and is restricted to a nonanticipative strategy $ \vu[\denoteLf] = \vStrategyNA[\denoteLf]\rBrack[{\vx,\vu[\denoteHf]}] $. \identPP\ employs the dynamics of the high-fidelity model and is restricted to a state-feedback strategy $ \vu[\denoteHf] = \vStrategyFB[\denoteHf]\rBrack[\vx] $. The relative state vector between \identPE\ and \identPP\ is given by \eqref{eq:rel_sys}, and the relative system dynamics between \identPE\ and \identPP\ are given by \eqref{eq:dyn_rel}. The initial state of the game is denoted as $ \vx\oftZero $.
\begin{definition}\label{def:scap}
    The so-called captivity set $ \Scap \subseteq \admissX $ is given by $ \Scap \coloneq \cBrack[{\vx\oft \in \admissX \vertM \Euclidean[{\vfselXcrit\rBrack[{\vx\oft}]}] \leq \sizeScap}] $, with $ \sizeScap \in \Reals[+] $.
\end{definition}

The boundary of the captivity set $ \Scap $ is denoted as $ \boundary\Scap $.
The objective function $ \objectiveFunc $ in a captivity-escape game is given by its terminal time $ \stimeEnd[\denoteEnd] $, $ \objectiveFunc = \stimeEnd[\denoteEnd] \coloneq \inf\cBrack[{\stime \vertM \vx\oft \not\in \Scap }] $. \identPP\ aims to maximize $ \objectiveFunc $, whereas \identPE\ aims to minimize $ \objectiveFunc $.
\identPE\ initializes in \textit{captivity}, and thus, the initial state fulfills $ \vx\oftZero \in \Scap $. In the game, \identPE\ seeks to \textit{escape} from $ \Scap $ by achieving any $ \vx\oft \notin \Scap $. Conversely, \identPP\ strives to retain \identPE\ in captivity by maintaining $ \vx\oft \in \Scap $. The game terminates when \identPE\ escapes $ \vx\oft \notin \Scap $.

Note the difference between a captivity-escape game and a pursuit-evasion game as employed in \autocite{herbert.2017}, \autocite{chen.2021}. In a pursuit-evasion game, a player $ \tilde{\text{\identPE}} $ is initially \textit{free}, and the game terminates when a second player $ \tilde{\text{\identPP}} $ \textit{captures} $ \tilde{\text{\identPE}} $ inside a closed target set $ \tilde{\Scap} $ {\autocite{patsko.2018}}. In contrast, in the formulated captivity-escape game, \identPE\ is initially in \textit{captivity}, and the game terminates when \identPE\ \textit{escapes} from the compact captivity set $ \Scap $. This shift in perspective, highlighted in \autoref{fig:two_games}, is central for efficiently addressing \ref{O.0}, \ref{O.1}, and \ref{O.2}.
\subsection{Formulation of a Captivity-Escape Game of Kind}\label{sec:cegok}
For solving Problem~\ref{problem_1} and Problem~\ref{problem_2}, it is sufficient to investigate under what conditions \identPP\ can prevent termination of the game. To this end, consider a captivity-escape game of kind with only two possible outcomes:
\begin{enumerate}[{o.}i]
    \item Player \identPP\ is able to retain eternal captivity despite \identPE's best effort to terminate the game by escaping.\label{out:captivity}
    \item Player \identPE\ is able to escape in finite time despite \identPP's best effort to retain \identPE\ in captivity.\label{out:escape}
\end{enumerate}

The objective function $ \objectiveFunc_\mathrm{k} $ in a captivity-escape game of kind can only take two values and is given by
\begin{equation}\label{eq:qual_obj}
    \objectiveFunc_\mathrm{k} \coloneq 
    \left\{
        \begin{array}{ll}
            +1 & \text{if } \vx\oft \in \Scap \text{, } \forall\stime\in [\stime[\denoteZero],\infty),\\
            -1 & \text{otherwise.}
        \end{array}
        \right.
\end{equation}
\identPP\ aims to maximize $ \objectiveFunc_\mathrm{k} $, whereas \identPE\ aims to minimize $ \objectiveFunc_\mathrm{k} $.
The related value function $ \valueFunc\rBrack[\vx\oft] $ for this game is given by
\begin{equation}\label{eq:qual_val}
    \valueFunc\rBrack[\vx\oft] \mkern-2mu=\mkern-2mu 
    \left\{
        \begin{array}{ll}
            \mkern-9mu +1 & \mkern-9mu \text{if optimal play from } \mkern-2mu\vx\oft\mkern-2mu \text{ yields o.\ref{out:captivity}},\\
            \mkern-9mu -1 & \mkern-9mu \text{if optimal play from } \mkern-2mu\vx\oft\mkern-2mu \text{ yields o.\ref{out:escape}}.
        \end{array}
        \right.
\end{equation}
\begin{definition}\label{def:captivity_zone}
    The set $ \Szcap \subseteq \Scap $, with $ \Szcap \coloneq \{\vx\oft \in\Scap \vertM $\\
    \noindent $ \valueFunc\rBrack[\vx\oft] = +1 \} $, forms the so-called captivity zone of $ \Scap $.
\end{definition}
\begin{definition}\label{def:escape_zone}
    The set $ \Szesc \subseteq \Scap $, with $ \Szesc \coloneq \{\vx\oft \in\Scap \vertM $\\
    \noindent $ \valueFunc\rBrack[\vx\oft] = -1 \} $, forms the inevitable escape zone of $ \Scap $, from which escape from $ \Scap $ occurs in finite time under optimal play.
\end{definition}
\begin{assumption}\label{assump:captivity_zone_exists}
    The captivity zone $ \Szcap $ of a captivity-escape game of kind is nonempty and compact.
\end{assumption}

The boundary of $ \Szcap $ is denoted as $ \boundary\Szcap $.
\begin{definition}\label{def:opt_strategies_game_of_kind}
    In a captivity-escape game of kind, the optimal strategies of \identPE\ and \identPP\ are denoted as $ \vu[\denoteLf] = \vStrategyNA[\denoteLf]^{\denoteGopt}\rBrack[{\vx,\vu[\denoteHf]}] $ and $ \vu[\denoteHf] = \vStrategyFB[\denoteHf]^{\denoteGopt}\rBrack[\vx] $, respectively. These optimal strategies are only defined for $ \vx\oft \in \boundary\Scap $ and $ \vx\oft \in \boundary\Szcap $.
\end{definition}

Consequently, if both players use their optimal strategy, trajectories starting in the captivity zone $ \Szcap $ stay in the captivity set $ \Scap $: $ \forall\vx\oftZero \in \Szcap \Rightarrow \forall \stime \geq \stime[\denoteZero] $: $ \vSolTrj(\stime;\stime[\denoteZero],\vx\oftZero,$
$\vStrategyNA[\denoteLf]^{\denoteGopt}\rBrack[{\vx,\vu[\denoteHf]}],\vStrategyFB[\denoteHf]^{\denoteGopt}\rBrack[{\vx}]) \in \Scap $. Thus, we say initializing the game in $ \vx\oftZero\in\Szcap $ results in eternal captivity.

%% file: 02_Img/figure_game_difference.tex
\begin{tikzpicture}[scale=0.9]
    \draw [draw=none] (0,0.2cm) rectangle (7cm,2.8cm);

    \node [black,anchor=center,font=\footnotesize] at (1.75cm,2.8cm) {Captivity-Escape Game};

    \draw [color=white, fill=black!10] (.2cm,.2cm) rectangle (3.3cm,2.6cm);
    \draw [color=mygreen, thick, fill=mygreenfill!20] (1.75cm,1.4cm) circle (1cm);
    \draw[black!80, very thin] (1.75cm,1.4cm) -- (\figGameDiffEndBetaX cm,\figGameDiffEndBetaY cm);
    \node [black,anchor=south east,font=\footnotesize] at (1.45cm,.7cm) {$ \Scap $};
    \node [black,anchor=north east,font=\footnotesize] at (1.45cm,1.8cm) {$ \sizeScap $};

    \draw [fill] (1.75cm,1.4cm) circle (.05cm);
    \node [black,anchor=south west,font=\footnotesize] at (1.75cm,1cm) {\identPP};
    
    \draw [fill] (2.15cm,1.75cm) circle (.05cm);
    \node [black,anchor=south west,font=\footnotesize] at (2.15cm,1.35cm) {\identPE};
    
    \draw[black!80, densely dashed] (3.5cm,3cm) -- (3.5cm,.2cm) ;

    \begin{scope}[shift={(3.5cm,0cm)}]
        \node [black,anchor=center,font=\footnotesize] at (1.75cm,2.8cm) {Pursuit-Evasion Game};
        \draw [color=white, fill=mygreenfill!20] (.2cm,.2cm) rectangle (3.3cm,2.6cm);
    
        \draw [color=black, thick,fill=black!10] (1.75cm,1.4cm) circle (1cm);
        \draw[black!80, very thin] (1.75cm,1.4cm) -- (\figGameDiffEndBetaX cm,\figGameDiffEndBetaY cm);
        \node [black,anchor=south east,font=\footnotesize] at (1.45cm,.7cm) {$ \tilde{\Scap} $};
        \node [black,anchor=north east,font=\footnotesize] at (1.45cm,1.8cm) {$ \sizeScap $};
    
        \draw [fill] (1.75cm,1.4cm) circle (.05cm);
        \node [black,anchor=south west,font=\footnotesize] at (1.75cm,1cm) {$ \tilde{\text{\identPP}} $};
    
        \draw [fill] (2.7cm,2.25cm) circle (.05cm);
        \node [black,anchor=south west,font=\footnotesize] at (2.7cm,1.85cm) {$ \tilde{\text{\identPE}} $};
    \end{scope}
\end{tikzpicture}

%% file: 01_Tex/04_game2quant_and_ctrl.tex
\section{From a Captivity-Escape Game of Kind to a \glsentryshort{TEB}, the Related \glsentryshort{WTE}, and a Safety Controller}\label{sec:04_relation}
In this section, we formulate the relation between a \gls*{TEB} $ \Steb $, a \gls*{WTE} $ \wte $, a safety controller, and a captivity-escape game of kind as formulated in \autoref{sec:cegok}.
\subsection{Captivity-Escape Game and a \glsentryshort{TEB}}\label{sec:EC_TEB}
\begin{remark}
    Open-loop trajectories of \eqref{eq:dyn_rel} are understood in the Carathéodory sense; due to possible discontinuities in $ \vStrategyFB[\denoteHf]\rBrack[{\vx}] $ and $ \vStrategyNA[\denoteLf]\rBrack[{\vx,\vu[\denoteHf]}] $, closed-loop trajectories $ \vSolTrj(\cdot;\stime[\denoteZero],\vx\oftZero,$
    $\vStrategyNA[\denoteLf](\vx,\vu[\denoteHf]),\vStrategyFB[\denoteHf]\rBrack[{\vx}]) $ are understood in the Filippov sense \autocite{cortes.2008}.
\end{remark}
\begin{lemma}\label{lemma:1}
    Let $ \vSolTrj(\cdot;\stime[\denoteZero],\vx\oftZero,\vStrategyNA[\denoteLf](\vx,\vu[\denoteHf]),\vStrategyFB[\denoteHf]\rBrack[{\vx}]) $ be a closed-loop solution of \eqref{eq:dyn_rel} on $ [\stime[\denoteZero],\stimeEnd] $, and let 
    $ \tilde{\vx} = \vSolTrj(\tilde{\stime};\stime[\denoteZero],\vx\oftZero,\vStrategyNA[\denoteLf](\vx,$
    $\vu[\denoteHf]),\vStrategyFB[\denoteHf]\rBrack[{\vx}]) $, $ \tilde{\stime} \in $ $ [\stime[\denoteZero],\stimeEnd] $. The restriction $ \tilde{\vSolTrj}(\cdot;\tilde{\stime},\tilde{\vx},\vStrategyNA[\denoteLf](\vx,\vu[\denoteHf]),$
    $\vStrategyFB[\denoteHf]\rBrack[{\vx}]): $ $ [\tilde{\stime},\stimeEnd] \rightarrow \admissX $, given by 
    $ \tilde{\vSolTrj}(\stime;\tilde{\stime},\tilde{\vx},\vStrategyNA[\denoteLf] \rBrack[{\vx,\vu[\denoteHf]}],\vStrategyFB[\denoteHf] \rBrack[{\vx}]) = $
    $ \vSolTrj(\stime;\stime[\denoteZero],\vx\oftZero,\vStrategyNA[\denoteLf]\rBrack[{\vx,\vu[\denoteHf]}],\vStrategyFB[\denoteHf]\rBrack[{\vx}]) $, is again a solution on $ [\tilde{\stime},\stimeEnd] $.
\end{lemma}
\begin{proof}
    Since $ \vSolTrj(\cdot;\stime[\denoteZero],\vx\oftZero,\vStrategyNA[\denoteLf](\vx,\vu[\denoteHf]),\vStrategyFB[\denoteHf]\rBrack[{\vx}]) $ is a closed-loop solution on $ [\stime[\denoteZero],\stimeEnd] $, it is absolutely continuous and satisfies \eqref{eq:dyn_rel} a.e. on $ [\stime[\denoteZero],\stimeEnd] $, and hence on any $ [\tilde{\stime},\stimeEnd] \subseteq [\stime[\denoteZero],\stimeEnd] $.
\end{proof}
\begin{theorem}\label{theo:TEB}
    Let Assumption~\ref{assump:captivity_zone_exists} hold, and for $ \vx\oft \notin \boundary\Scap \cup \boundary\Szcap $, let $ \vStrategyNA[\denoteLf]^{\denoteGopt}\rBrack[{\vx,\vu[\denoteHf]}] $ and $ \vStrategyFB[\denoteHf]^{\denoteGopt}\rBrack[\vx] $ return any $ \vu[\denoteLf]\oft \in \admissU[\denoteLf] $ and $ \vu[\denoteHf]\oft \in \admissU[\denoteHf] $, respectively.
    The captivity zone $ \Szcap $ is robust positively invariant, thus forming a \gls*{TEB} $ \Steb $ constituted by means of $ \vu[\denoteHf] = \vStrategyFB[\denoteHf]^{\denoteGopt}\rBrack[\vx] $, i.e., for any $ \vx\oftZero \in \Szcap $, $ \vSolTrj(\stime;\stime[\denoteZero], \vx\oftZero,$
    $ \vStrategyNA[\denoteLf]^{\denoteGopt}\rBrack[{\vx,\vu[\denoteHf]}],\vStrategyFB[\denoteHf]^{\denoteGopt}\rBrack[{\vx}]) \in \Szcap $, $ \forall \stime \in [\stime[\denoteZero],\stimeEnd] $.
\end{theorem}
\begin{proof}
    Let $ \vx\oftZero \in \Szcap $ and suppose, for contradiction, that there exists $ \tilde{\stime} \in [\stime[\denoteZero],\stimeEnd] $ with $ \tilde{\vx} = \vSolTrj(\tilde{\stime};\stime[\denoteZero],\vx\oftZero,\vStrategyNA[\denoteLf]^{\denoteGopt}(\vx,\vu[\denoteHf]),$ 
    $\vStrategyFB[\denoteHf]^{\denoteGopt}\rBrack[{\vx}]) \in \Szesc $. By Definition~\ref{def:escape_zone}, the restricted trajectory $ \tilde{\vSolTrj}(\cdot;\tilde{\stime},\tilde{\vx},\vStrategyNA[\denoteLf]^{\denoteGopt}\rBrack[{\vx,\vu[\denoteHf]}],\vStrategyFB[\denoteHf]^{\denoteGopt}\rBrack[{\vx}]) $, which is a closed-loop solution of \eqref{eq:dyn_rel} by Lemma~\ref{lemma:1}, implies termination of the game by escape, which contradicts $ \valueFunc\rBrack[\vx\oftZero] = +1 $ (see \eqref{eq:qual_val}). Thus, $ \vSolTrj(\stime;\stime[\denoteZero], \vx\oftZero, \vStrategyNA[\denoteLf]^{\denoteGopt}\rBrack[{\vx,\vu[\denoteHf]}],\vStrategyFB[\denoteHf]^{\denoteGopt}\rBrack[{\vx}]) \in \Szcap $, $ \forall \stime \in [\stime[\denoteZero],\stimeEnd] $.
\end{proof}
\subsection{Captivity-Escape Game and a \glsentryshort{WTE}}
\begin{theorem}\label{theo:zeta_beta}
    Let Assumption~\ref{assump:captivity_zone_exists} hold. The size $ \sizeScap $ of the captivity set $ \Scap $ is an upper bound of the \gls*{WTE} $ \wte $ related to the \gls*{TEB} $ \Steb $ formed by $ \Szcap $ (see Theorem~\ref{theo:TEB}): $ \wte \leq \sizeScap $.
\end{theorem}
\begin{proof}
    The largest possible $ \hat{\wte} $ corresponds to the largest possible $ \hat{\Steb} $, which is given by the largest possible $ \Szcapmax $. As $ \Szcap \subseteq \Scap $, $ \hat{\Steb} = \Scap $ and
    $ \hat{\wte} = \max_{\vx\oft \in \hat{\Steb}} \Euclidean[{\vfselXcrit\rBrack[\vx\oft]}] = \max_{\vx\oft \in \Scap} \Euclidean[{\vfselXcrit\rBrack[\vx\oft]}] = \sizeScap $ (see Definition~\ref{def:wte}).
\end{proof}
\begin{corollary}\label{cor:alpha_beta}
    The size $ \sizeScap $ of the captivity set $ \Scap $ serves as a safety margin $ \ssafetyMargin $: $ \wte \leq \sizeScap = \ssafetyMargin $ (see Remark~\ref{rem:smargin}).
\end{corollary}
\begin{assumption}\label{assump:capZcontainsBoundary}
    The captivity zone $ \Szcap $ contains at least one state on the boundary of the captivity set $ \boundary\Scap $: $ \Szcap \cap \boundary\Scap \neq \emptyset $.
\end{assumption}
\begin{corollary}\label{cor:zeta_beta}
    Let Assumption~\ref{assump:captivity_zone_exists} and Assumption~\ref{assump:capZcontainsBoundary} hold: $ \wte = \sizeScap = \ssafetyMargin $ holds (see Theorem~\ref{theo:zeta_beta}, Corollary~\ref{cor:alpha_beta}).
\end{corollary}

By Corollary~\ref{cor:zeta_beta}, a captivity-escape game of kind can be used to solve Problem~\ref{problem_2} with $ \wte = \ssafetyMargin $ under Assumptions~\ref{assump:captivity_zone_exists} and \ref{assump:capZcontainsBoundary}.
\subsection{Captivity-Escape Game and a Safety Controller}
\begin{theorem}\label{theo:safety_controller}
    Let Assumption~\ref{assump:captivity_zone_exists} hold. The strategy
    \begin{equation}
        \vSCtrlLaw\rBrack[\vx\oft] \coloneq
        \left\{
        \begin{array}{ll}
            \vStrategyFB[\denoteHf]^{\denoteGopt}\rBrack[\vx\oft] & \text{if } \vx\oft \in \boundary\Steb,\\
            \text{any }\vu[\denoteHf]\oft \in \admissU[\denoteHf] & \text{if } \vx\oft \in \mathrm{int}(\Steb), 
        \end{array}
        \right. 
    \end{equation}
    ensures that for any $ \vx\oftZero \in \Steb $ and any $ \vu[\denoteLf] = \vStrategyNA[\denoteLf]\rBrack[{\vx,\vu[\denoteHf]}] $: $ \vSolTrj(\stime;\stime[\denoteZero],\vx\oftZero,\vStrategyNA[\denoteLf]\rBrack[{\vx,\vu[\denoteHf]}],$ $\vSCtrlLaw\rBrack[\vx]) \in \Steb $ holds $ \forall \stime \geq \stime[\denoteZero] $.
\end{theorem}
\begin{proof}
    To leave $ \Steb $, $ \vSolTrj(\stime;\stime[\denoteZero],\vx\oftZero,\vStrategyNA[\denoteLf]\rBrack[{\vx,\vu[\denoteHf]}],\vSCtrlLaw\rBrack[\vx]) $ with $ \vx\oftZero \in \Steb $ must pass a state $ \vx\oft \in \boundary\Steb $. Thus, it is sufficient to show that $ \vSCtrlLaw\rBrack[\vx] $ ensures for any $ \vx\oftZero \in \boundary\Steb $ and any $ \vu[\denoteLf] = \vStrategyNA[\denoteLf]\rBrack[{\vx,\vu[\denoteHf]}] $: $ \vSolTrj(\stime;\stime[\denoteZero],\vx\oftZero,\vStrategyNA[\denoteLf]\rBrack[{\vx,\vu[\denoteHf]}],\vSCtrlLaw\rBrack[\vx]) \in \Steb $ holds $ \forall \stime \geq \stime[\denoteZero] $. This holds by Definition~\ref{def:opt_strategies_game_of_kind} and Theorem~\ref{theo:TEB}, as any $ \vStrategyNA[\denoteLf]\rBrack[{\vx,\vu[\denoteHf]}] $ is equally optimal or suboptimal compared to $ \vStrategyNA[\denoteLf]^{\denoteGopt}\rBrack[{\vx,\vu[\denoteHf]}] $ w.r.t. forcing a trajectory to leave $ \Steb $.
\end{proof}
\begin{corollary}\label{cor:safety_controller}
    The \gls*{TEB} $ \Steb $ is robust positively invariant for \eqref{eq:dyn_rel} with the minimal-intervention safety controller $ \vu[\denoteHf] = \vSCtrlLaw\rBrack[\vx] $.
\end{corollary}
\begin{corollary}
    The strategies $ \vStrategyNA[\denoteLf]^\bullet\rBrack[{\vx,\vu[\denoteHf]}] $ and $ \vStrategyFB[\denoteHf]^\bullet\rBrack[\vx] $ in Problem~\ref{problem_1} are given by $ \vStrategyNA[\denoteLf]^{\denoteGopt}\rBrack[{\vx,\vu[\denoteHf]}] $ and $ \vStrategyFB[\denoteHf]^{\denoteGopt}\rBrack[\vx] $, for $ \vx\oft \in \boundary\Steb $.
\end{corollary}

%% file: 01_Tex/05_quant_and_ctrl.tex
\section{Addressing \ref{O.0}, \ref{O.1}, and \ref{O.2} by Means of a Captivity-Escape Game of Kind}\label{sec:05_solving}
In this section, we present a method that determines a \gls*{TEB} $ \Steb $ through a captivity-escape game of kind. To this end, we outline an existing approach for addressing zero-sum differential games of kind in \autoref{sec:preliminaries} and complement this approach in \autoref{sec:analyzing_cegc} to address \ref{O.0}, \ref{O.1}, and \ref{O.2}.
\subsection{Addressing a Captivity-Escape Game of Kind with a Given Planning Performance and Safety Bound}\label{sec:preliminaries}
In this section, we outline the approach for addressing games of kind initially presented in \autocite[Chap. 8]{isaacs.1999} and further investigated in \autocite{patsko.2018}, \autocite{buzikov.2023}. This approach determines a compact invariant region $ \SCRegion \subseteq \Szcap $ by constructing the boundary $ \boundary\SCRegion $ using the so-called \gls*{IP} of $ \boundary\Scap $ and a closed barrier $ \Sbarr $. While the existing approach aims to determine $ \SCRegion = \Szcap $, this equality cannot be ensured in general. Nevertheless, the resulting region $ \SCRegion $ is guaranteed to retain the properties of $ \Szcap $ that are relevant for ensuring safety. Within this section, assume a suitable planning performance $ \pp $ and safety margin $ \ssafetyMargin $ to be given that ensure the existence of $ \Szcap $.
\subsubsection{The Inward-Facing Part of $ \boundary\Scap $}
Let $ \vonScap \in \Reals[\sizeX] $ denote the outward normal to the captivity set $ \Scap $ at a state $ \vx\oft \in \boundary\Scap $. If \identPP\ acts optimally, \identPE\ cannot terminate the game immediately by escaping from a state $ \vx\oft \in \boundary\Scap $ where
\begin{equation}\label{eq:def_sp}
    \min\limits_{\vu[\denoteHf]\oft\in\admissU[\denoteHf]}\max\limits_{\vu[\denoteLf]\oft\in\admissU[\denoteLf]}\vonScap^{\Transpose} \reldyn \leq 0.
\end{equation}
\glsreset{IP}
\begin{definition}\label{def:sp}
    The states $ \vx\oft \in \boundary\Scap $ satisfying \eqref{eq:def_sp} constitute the so-called \gls*{IP} of $ \boundary\Scap $.\footnote{Game-theoretic terminology refers to these states as the \textit{nonusable part} of $ \partial\Scap $, since optimal play by both players does not lead to immediate termination of the game (see \autocite[{Sec. 4.7}]{isaacs.1999}). However, precisely because the game does not terminate immediately, this part of $ \partial\Scap $ is \textit{usable} for constructing the captivity zone $ \Szcap $. To avoid potential misinterpretation, we therefore use the neutral term \textit{inward-facing part}.}
\end{definition}
\begin{definition}\label{def:bsp}
    The states $ \vx\oft \in \boundary\Scap $ satisfying equality in \eqref{eq:def_sp} constitute the so-called \gls*{BIP} that is denoted as $ \SBIP \subseteq \boundary\Scap $.
\end{definition}
\begin{remark}\label{rem:kappa}
    By solving \eqref{eq:def_sp} for equality, the \gls*{BIP} $ \SBIP $ can be determined by means of a parameter vector $ \vparamBIP \in \Reals[\sizeX-2] $ through a mapping $ \mapBIP(\vparamBIP) $, $ \mapBIP: \Reals[\sizeX-2] \rightarrow \SBIP $.
\end{remark}

Consider the optimal strategies w.r.t. \eqref{eq:def_sp} to be denoted as $ \vu[\denoteHf] = \vStrategyFB[\denoteHf]^\vartriangle\rBrack[\vx] $ and $ \vu[\denoteLf] = \vStrategyNA[\denoteLf]^\vartriangle(\vx,\vu[\denoteHf]) $. Definition~\ref{def:sp} indicates that a solution trajectory of \eqref{eq:dyn_rel} $ \vSolTrj(\stime;\stime[\denoteZero], \vx\oftZero, \vStrategyNA[\denoteLf](\vx,\vu[\denoteHf]),\vStrategyFB[\denoteHf]^\vartriangle\rBrack[\vx]) $ cannot cross the \gls*{IP} for any $ \vu[\denoteLf] = \vStrategyNA[\denoteLf](\vx,\vu[\denoteHf]) $. Thus, the \gls*{IP} qualifies for constructing a part of $ \boundary\Szcap $, and thus of $ \boundary\SCRegion $.
\subsubsection{Semipermeable Surfaces $ \Ssemsur $ and Closed Barriers $ \Sbarr $}
Recall the goal of determining $ \SCRegion = \Szcap $. A surface separating $ \Szcap $ and $ \Szesc $ must not be crossed during optimal play. Therefore, such a surface must be semipermeable \autocite[Sec. 8.5]{isaacs.1999}.
\begin{definition}[{\autocite[Sec. 2.3]{patsko.2018}}]\label{def:sem_sur}
    Let $ \vonBarr \in \Reals[\sizeX] $ denote a nonzero outward normal to a smooth surface $ \Ssemsur \subseteq \Reals[\sizeX] $ at $ \vx\oft \in \Ssemsur $. A surface $ \Ssemsur $ is semipermeable, if for all $ \vx\oft \in \Ssemsur $
    \begin{equation}\label{eq:def_sem_sur}
        \min\limits_{\vu[\denoteHf]\oft \in \admissU[\denoteHf]}\max\limits_{\vu[\denoteLf]\oft \in \admissU[\denoteLf]} \vonBarr ^{\Transpose} \reldyn = 0.
    \end{equation}
\end{definition}
\vspace{.4em}
\begin{definition}\label{def:closed_barrier}
    One or more $ \Ssemsur $ that jointly delimit a compact region $ \SCRegion \subseteq \Scap $ constitute a so-called closed barrier $ \Sbarr $.
\end{definition}
\subsubsection{Construction of a Closed Barrier $ \Sbarr $}
To determine a region $ \SCRegion \subseteq \Szcap $, the approach in \autocite[Sec. 8.5]{isaacs.1999} aims to construct a closed barrier $ \Sbarr $ by means of semipermeable surfaces $ \Ssemsur $ that smoothly connect to the \gls*{IP} at the \gls*{BIP} $ \SBIP $.

Consider the optimal strategies w.r.t. \eqref{eq:def_sem_sur} to be denoted as $ \vu[\denoteLf] = \vStrategyNA[\denoteLf]^{\denoteOpt}\rBrack[{\vx,\vu[\denoteHf]}] $ and $ \vu[\denoteHf] = \vStrategyFB[\denoteHf]^{\denoteOpt}\rBrack[{\vx}] $. Substituting $ \vStrategyNA[\denoteLf]^{\denoteOpt}\rBrack[{\vx,\vu[\denoteHf]}] $ and $ \vStrategyFB[\denoteHf]^{\denoteOpt}\rBrack[{\vx}] $ into \eqref{eq:def_sem_sur} yields the identity
\begin{equation}\label{eq:xi_}
    \vonBarr^{\Transpose} {\vf^{\vmp}\rBrack[{\vx\oft,\vStrategyNA[\denoteLf]^{\denoteOpt}\rBrack[{\vx\oft,\vStrategyFB[\denoteHf]^{\denoteOpt}\rBrack[{\vx\oft}]}], \vStrategyFB[\denoteHf]^{\denoteOpt}\rBrack[{\vx\oft}]}]} \equiv 0.
\end{equation}
Differentiation of \eqref{eq:xi_} w.r.t. $ \vx $ results in
\begin{equation}\label{eq:xi_diff}
    \scalebox{0.8}{$
        \dFull{\vonBarr}{\stime} = -\rBrack[{\dPartial{}{\vx} \vf^{\vmp}\rBrack[{\vx\oft, \vStrategyNA[\denoteLf]^{\denoteOpt}\rBrack[{\vx\oft,\vStrategyFB[\denoteHf]^{\denoteOpt}\rBrack[{\vx\oft}]}], \vStrategyFB[\denoteHf]^{\denoteOpt}\rBrack[{\vx\oft}]}]}]^{\Transpose} \vonBarr
    $},
\end{equation}
which is evaluated along the semipermeable surfaces $ \Ssemsur $ using 
\begin{equation}\label{eq:xi_cond}
    \vonBarr = \vonScap
\end{equation}
at $ \stime = \stimeBip $ to ensure a smooth connection to the \gls*{IP}. The computations resulting in \eqref{eq:xi_diff} exploit the optimality of $ \vStrategyFB[\denoteHf]^{\denoteOpt}\rBrack[{\vx\oft}] $ and $ \vStrategyNA[\denoteLf]^{\denoteOpt}\rBrack[{\vx\oft,\vu[\denoteHf]\oft}] $ w.r.t. \eqref{eq:def_sem_sur} (for details, see \autocite[Sec. 8.3]{isaacs.1999}).
By using the solution of \eqref{eq:xi_diff} and \eqref{eq:xi_cond}, semipermeable surfaces $ \Ssemsur $ that connect to the \gls*{BIP} are constructed by solving
\begin{equation}\label{eq:rel_dyn_barr}
    \vxdot\oft = \vf^{\vmp}\rBrack[{\vx\oft,\vStrategyNA[\denoteLf]^{\denoteOpt}\rBrack[{\vx\oft,\vu[\denoteHf]\oft}], \vStrategyFB[\denoteHf]^{\denoteOpt}\rBrack[{\vx\oft}]}]
\end{equation}
in retrograde time for all
\begin{equation}\label{eq:x_cond}
    \vx\rBrack[\stimeBip] \in \SBIP,
\end{equation}
which results in a set of trajectories
\begin{equation}\label{eq:sol_trajectories}
    \begin{aligned}
        & \vSolTrj (\stime;\stime[\denoteZero],\vx\oftZero,\vStrategyNA[\denoteLf]^{\denoteOpt}\rBrack[{\vx,\vu[\denoteHf]}], \vStrategyFB[\denoteHf]^{\denoteOpt}\rBrack[{\vx}] ),\text{ with } \vx\oftZero \in \Ssemsur,\\
        & \stime[\denoteZero] < \stimeBip, \text{ and } \vSolTrj(\stimeBip;\stime[\denoteZero],\vx\oftZero,\vStrategyNA[\denoteLf]^{\denoteOpt}\rBrack[{\vx,\vu[\denoteHf]}], \vStrategyFB[\denoteHf]^{\denoteOpt}\rBrack[{\vx}])\in \SBIP,
    \end{aligned}
\end{equation}
that can be parametrized by means of $ \vparamBIP $ (see Remark~\ref{rem:kappa}).

If one or more of these semipermeable surfaces $ \Ssemsur $ intersect with another $ \Ssemsur $ or the \gls*{IP} in a nonleaking manner (see \cite[Sec. 4.3]{lewin.1994}), they jointly compose a closed barrier $ \Sbarr $, thus yielding a compact region $ \SCRegion \subseteq \Scap $, with $ \partial\SCRegion $ composed of $ \Sbarr $ and the \gls*{IP}. If a semipermeable surface $ \Ssemsur $ intersects itself, the portions beyond the intersection are redundant \autocite{buzikov.2023}, \autocite[Sec. 8.5]{isaacs.1999}. 

As \identPE\ can neither force $ \vx\oft \in \SCRegion $ through the barrier $ \Sbarr $ nor through the \gls*{IP}, $ \Sbarr \subseteq \SCRegion \subseteq \Szcap $ holds. Thus, $ \SCRegion $ retains the properties of $ \Szcap $ relevant for ensuring safety by means of $ \vu[\denoteHf] = \vStrategyFB[\denoteHf]^\vartriangle\rBrack[\vx] $ for $ \vx\oft \in $ \gls*{IP}, and $ \vu[\denoteHf] = \vStrategyFB[\denoteHf]^{\denoteOpt}\rBrack[{\vx}] $ for $ \vx\oft \in \Sbarr $. 

Note that by connecting the semipermeable surfaces to the \gls*{BIP}, Assumption~\ref{assump:capZcontainsBoundary} holds automatically.
\subsection{Addressing a Captivity-Escape Game of Kind with a Variable Planning Performance and Safety Bound}\label{sec:analyzing_cegc}
In contrast to \autoref{sec:preliminaries}, consider both $ \pp $ and $ \ssafetyMargin $ to be variable within this section. For addressing \ref{O.1}, we aim to determine $ \pp $ such that the semipermeable surfaces $ \Ssemsur $ resulting from \eqref{eq:rel_dyn_barr} and \eqref{eq:x_cond} jointly constitute a closed barrier $ \Sbarr $ that yields $ \SCRegion \subseteq \Szcap $ of a captivity-escape game of kind with $ \sizeScap = \ssafetyMargin $.

To highlight the dependency on $ \pp $ and $ \sizeScap $, we denote the trajectories in \eqref{eq:sol_trajectories} as $ \vSolTrj[\pp,\sizeScap]^{\Ssemsur} $. Moreover, we denote the part of $ \vSolTrj[\pp,\sizeScap]^{\Ssemsur} $ that contributes to constituting $ \Sbarr $ as $ \vSolTrj[\pp,\sizeScap]^{\Sbarr} $: a part $ \vSolTrj[\pp,\sizeScap]^{\Sbarr} $ starts in an intersection with another $ \vSolTrj[\pp,\sizeScap]^{\Sbarr} $ or the \gls*{IP} at $ \stime = \hat{\stime} $ and ends when reaching the \gls*{BIP} at $ \stime = \stimeBip $. Note that the initial time $ \hat{\stime} $ depends on the particular $ \vx(\stimeBip) \in \SBIP $ related to the respective $ \vSolTrj[\pp,\sizeScap]^{\Sbarr} $. Thus, the respective $ \hat{t} $ can be determined by means of a mapping $ b(\vparamBIP) $, $ b: \Reals[\sizeX-2] \rightarrow \Reals[<\stimeBip] $ (see Remark~\ref{rem:kappa}).

To potentially contain a part $ \vSolTrj[\pp,\sizeScap]^{\Sbarr} $, trajectories $ \vSolTrj[\pp,\sizeScap]^{\Ssemsur} $ must approach the \gls*{BIP} $ \SBIP $ from the interior of $ \Scap $.
\begin{proposition}\label{prop:dderrmeas}
    Trajectories $ \vSolTrj[\pp,\sizeScap]^{\Ssemsur} $ approach the \gls*{BIP} $ \SBIP $ from the interior of the captivity set $ \Scap $, if
    \begin{equation}\label{eq:dderrmeas}
        \mathrm{d}^2/\mathrm{d}\stime^2\, \Euclidean[{\vfselXcrit\rBrack[{\vSolTrj[\pp,\sizeScap]^{\Ssemsur}}]}] \big\vert_{\stime = \stimeBip} < \vec{0}.
    \end{equation}
\end{proposition}
\begin{proof}
    \begin{equation}\label{eq:derrmeas}
        \mathrm{d}/\mathrm{d}\stime\, \Euclidean[{\vfselXcrit\rBrack[{\vSolTrj[\pp,\sizeScap]^{\Ssemsur}}]}] \big\vert_{\stime = \stimeBip} = \vec{0}
    \end{equation}
    holds by Definition~\ref{def:bsp}. If both \eqref{eq:dderrmeas} and \eqref{eq:derrmeas} hold, $ \stimeBip $ is a strict local maximizer of $ \Euclidean[{\vfselXcrit(\vSolTrj[\pp,\sizeScap]^{\Ssemsur})}] $ \autocite[Theorem 2.4]{nocedal.2006} with the maximum value $ \Euclidean[{\vfselXcrit(\vSolTrj[\pp,\sizeScap]^{\Ssemsur})}] \big\vert_{\stime = \stimeBip} = \sizeScap $. Thus, $ \Euclidean[{\vfselXcrit(\vSolTrj[\pp,\sizeScap]^{\Ssemsur})}] < \sizeScap $ holds in a punctured neighborhood of $ \stimeBip $, so $ \vSolTrj[\pp,\sizeScap]^{\Ssemsur} $ approaches the \gls*{BIP} $ \SBIP $ from the interior of $ \Scap $, and also leaves $ \SBIP $ into $ \Scap $.
\end{proof}

Moreover, Definition~\ref{def:closed_barrier} requires $ \vSolTrj[\pp,\sizeScap]^{\Sbarr} \subseteq \Scap $, which holds by Definition~\ref{def:scap} if 
\begin{equation}\label{eq:barr_inside_scap}
    \Euclidean[{\vfselXcrit\rBrack[{\vSolTrj[\pp,\sizeScap]^{\Sbarr}}]}] \leq \sizeScap, \quad \hat{\stime} \leq \stime \leq \stimeBip.
\end{equation}

Furthermore, to ensure $ \SCRegion $ is compact (see Assumption~\ref{assump:captivity_zone_exists}), trajectories $ \vSolTrj[\pp,\sizeScap]^{\Ssemsur} $ that contain a part $ \vSolTrj[\pp,\sizeScap]^{\Sbarr} $ must not leave $ \SCRegion $ immediately after passing the \gls*{BIP} $ \SBIP $.
\begin{assumption}\label{assump:constitutes_barr}
    $ \vSolTrj[\pp,\sizeScap]^{\Ssemsur} $ contains a part $ \vSolTrj[\pp,\sizeScap]^{\Sbarr} $.
\end{assumption}
\begin{proposition}\label{prop:goes_int}
    Let Assumption~\ref{assump:constitutes_barr} hold. If 
    \begin{equation}\label{eq:goes_int}
        \mathrm{d}/\mathrm{d}\stime\, \vSolTrj[\pp,\sizeScap]^{\Ssemsur} \big\vert_{\stime = \stimeBip} \text{ points towards the \gls*{IP}}, 
    \end{equation}
    and \eqref{eq:dderrmeas} holds, $ \vSolTrj[\pp,\sizeScap]^{\Ssemsur} $ leaves the \gls*{BIP} $ \SBIP $ into $ \SCRegion $.
\end{proposition}
\begin{proof}
    If Assumption~\ref{assump:constitutes_barr} holds, $ \SBIP $ separates the states on $ \boundary\Scap $ leading to immediate escape from the states resulting in eternal captivity, i.e., the \gls*{IP} (see Definition~\ref{eq:def_sp}).
\end{proof}
\begin{corollary}\label{cor:kappSzcap}
    Let Assumption~\ref{assump:constitutes_barr} hold. For any $ \vu[\denoteLf] = \vStrategyNA[\denoteLf]\rBrack[{\vx,\vu[\denoteHf]}] $, a trajectory $ \vSolTrj(\stime;\stime[\denoteZero],\vx\oftZero,\vStrategyNA[\denoteLf]\rBrack[{\vx,\vu[\denoteHf]}],\vStrategyFB[\denoteHf]^{\denoteOpt}\rBrack[{\vx}]) $ that solves \eqref{eq:dyn_rel} with $ \vx\oftZero \in \Sbarr $ and $ \stime[\denoteZero] < \stimeBip $ reaches the interior of $ \SCRegion $ at the latest after passing $ \SBIP $, if \eqref{eq:dderrmeas} is fulfilled: $ \Sbarr \subseteq \SCRegion $.
\end{corollary}
\begin{proof}
    The worst-case trajectory $ \vSolTrj(\stime;\stime[\denoteZero],\vx\oftZero,\vStrategyNA[\denoteLf]^{\denoteOpt}(\vx,$
    $\vu[\denoteHf]),\vStrategyFB[\denoteHf]^{\denoteOpt}\rBrack[{\vx}]) $ with $ \vx\oftZero \in \Sbarr $ and $ \stime[\denoteZero] < \stimeBip $ reaches the interior of $ \SCRegion $ after passing the \gls*{BIP} $ \SBIP $ at $ \stime = \stimeBip $ (see Proposition~\ref{prop:goes_int}). Thus, any suboptimal $ \vStrategyNA[\denoteLf]\rBrack[{\vx,\vu[\denoteHf]}] $ results in a trajectory reaching the interior of $ \SCRegion $ at $ \stime < \stimeBip $.
\end{proof}
\begin{corollary}\label{cor:vmin_is_compact}
    If $ \SCRegion \subseteq \Szcap $ is nonempty, it is compact.
\end{corollary}
\begin{proof}
    As $ \Scap $ is compact (see Definition~\ref{def:scap}), $ \SCRegion \subseteq \Szcap \subseteq \Scap $ is bounded. Moreover, $ \boundary\SCRegion $ is composed of the \gls*{IP} and a closed barrier $ \Sbarr $. If $ \SCRegion $ is nonempty, $ \Sbarr \subseteq \SCRegion $ (see Corollary~\ref{cor:kappSzcap}) and \gls*{IP} $\subseteq \SCRegion $ (see Definition~\ref{def:sp}). Thus, $ \SCRegion $ is closed.
\end{proof}

Depending on the employed low-fidelity and high-fidelity models, \eqref{eq:dderrmeas}, \eqref{eq:barr_inside_scap}, and \eqref{eq:goes_int} can only be fulfilled for a certain range of $ \sizeScap \geq \sizeScap_{\min} $ and $ \pp \leq \pp_{\max} $. The respective limits can be determined by solving \eqref{eq:dderrmeas}, \eqref{eq:barr_inside_scap}, and \eqref{eq:goes_int}. 

Consider the manifold in which the parts $ \vSolTrj[\pp,\sizeScap]^{\Sbarr} $ intersect with another $ \vSolTrj[\pp,\sizeScap]^{\Sbarr} $ or the \gls*{IP} for jointly composing a closed barrier $ \Sbarr $, to be denoted as $ \Smanifold \subseteq \Scap $. For a given $ \ssafetyMargin = \sizeScap \geq \sizeScap_{\min} $, there can be various manifolds $ \Smanifold $ in which the parts $ \vSolTrj[\pp,\sizeScap]^{\Sbarr} $ may intersect for jointly composing a closed barrier $ \Sbarr $ when fulfilling \eqref{eq:dderrmeas}, \eqref{eq:barr_inside_scap}, \eqref{eq:goes_int}. Each of these manifolds $ \Smanifold $ relates to a different planning performance $ \pp $, and thus, there is a mapping $ \vcondSol(\pp,\sizeScap) $, $ \vcondSol: \Reals[+,\leq\pp_\mathrm{max}] \times \Reals[\geq \sizeScap_\mathrm{min}] \rightarrow \powerset(\Smanifold) $, where $ \powerset(\Smanifold) $ denotes the power set of $ \Smanifold $.

Of particular interest is the mapping $ \overline{\vcondSol}(\sizeScap) $, $ \overline{\vcondSol}: \Reals[\geq \sizeScap_\mathrm{min}] \rightarrow \powerset(\Smanifold) $ that relates the optimal $ \pp $ to a given $ \sizeScap \geq \sizeScap_\mathrm{min} $ via
\begin{equation}\label{eq:manifold}
    \vSolTrj[\pp,\sizeScap]^{\Sbarr} \big\vert_{\stime = \hat{\stime}} \in \overline{\vcondSol}(\sizeScap).
\end{equation}
For a specific pair of low-fidelity and high-fidelity models, the mapping $ \overline{\vcondSol}(\sizeScap) $ can be determined from offline optimizing $ \pp $ subject to \eqref{eq:dderrmeas}, \eqref{eq:barr_inside_scap}, \eqref{eq:goes_int}, and \eqref{eq:manifold} for different $ \sizeScap \geq \sizeScap_{\min} $.
\begin{remark}\label{rem:powerful}
    If the low-fidelity model is powerful enough to always remain within some distance from the high-fidelity model, $ \Szcap $ exists \autocite{herbert.2017}, \autocite{chen.2021}. Thus, $ \overline{\vcondSol}(\sizeScap) $ exists for a proper pair of low-fidelity and high-fidelity models, if $ \SBIP $ is nonempty and $ \pp $ enables sufficiently restricting $ \admissX[\denoteLf] $, $ \admissU[\denoteLf] $, and $ \vmp[\denoteLf] $.
\end{remark}
\begin{remark}
    If desired, additional requirements can be considered by means of suitable objectives or constraints when determining $ \overline\vcondSol(\sizeScap) $ (e.g., requiring $ \SCRegion $ to be a connected set). 
\end{remark}
\subsubsection{Addressing \ref{O.1} and \ref{O.0}}
Given $ \ssafetyMargin = \sizeScap \geq \sizeScap_\mathrm{min} $, \ref{O.1} is addressed by solving \eqref{eq:dderrmeas}, \eqref{eq:barr_inside_scap}, \eqref{eq:goes_int}, and \eqref{eq:manifold} for $ \pp $. Similarly, \ref{O.0} is addressed by solving the respective equations for $ \ssafetyMargin = \sizeScap $, given $ \pp \leq \pp_\mathrm{\max} $ (or suitable $ \admissX[\denoteLf] $, $ \admissU[\denoteLf] $, and $ \vmp[\denoteLf] $).
\subsubsection{Addressing \ref{O.2}}
To address \ref{O.2}, the strategy $ \vStrategyFB[\denoteHf]^{\denoteGopt}\rBrack[\vx] $ used in the minimal-intervention safety controller must be determined (see Theorem~\ref{theo:safety_controller} and Corollary~\ref{cor:safety_controller}).
\begin{proposition}
    The optimal strategy $ \vStrategyFB[\denoteHf]^{\denoteGopt}\rBrack[\vx] $ employed in Theorem~\ref{theo:safety_controller} and Corollary~\ref{cor:safety_controller} is given by
    \begin{equation}
        \vStrategyFB[\denoteHf]^{\denoteGopt}\rBrack[\vx\oft] = 
        \left\{
            \begin{array}{ll}
                \vStrategyFB[\denoteHf]^\vartriangle\rBrack[\vx\oft] & \text{if } \vx\oft \in \text{\gls*{IP}} \cap \boundary\SCRegion,\\
                \vStrategyFB[\denoteHf]^{\denoteOpt}\rBrack[\vx\oft] & \text{if } \vx\oft \in \Sbarr \subseteq \boundary\SCRegion.
            \end{array}
            \right.
        \end{equation}
\end{proposition}
\vspace{.5em}
\begin{proof}
    The outlined method constructs $ \boundary\SCRegion $ utilizing the \gls*{IP} and $ \Sbarr $. For a state $ \vx\oft \in $ \gls*{IP}$ \,\cap\, \boundary\SCRegion $, the state-feedback strategy $ \vu[\denoteHf] = \vStrategyFB[\denoteHf]^\vartriangle\rBrack[\vx] $ ensures $ \vonScap^{\Transpose} \reldyn \leq 0 $ for any $ \vu[\denoteLf]\oft $ (see Definition~\ref{def:sp}). Similarly, the strategy $ \vu[\denoteHf] = \vStrategyFB[\denoteHf]^{\denoteOpt}\rBrack[{\vx}] $ ensures that any trajectory starting in $ \vx\oft \in \Sbarr \subseteq \boundary\SCRegion $ is steered into the interior of $ \SCRegion \subseteq \Szcap $ (see Corollary~\ref{cor:kappSzcap}).
\end{proof}

%% file: 01_Tex/06_numerical_example.tex
\section{Numerical Example and\\ Comparison to \glsentryshort{fastrack}}\label{sec:06_numerical_example}
\begin{figure*}[!t]
    \centering
    \begin{subfigure}[t]{0.329\textwidth}
        \captionsetup{margin=5pt}
        \centering
        \includegraphics{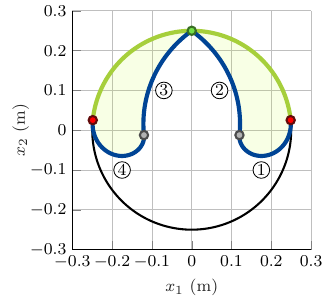}
        \caption{The \gls*{TEB} that results from addressing \ref{O.1} using the presented method. The semipermeable surfaces $ \Ssemsur $ intersect at $ \overline\vcondSol(\sizeScap) = \big\{\begin{bmatrix} 0 & \SI{0.25}{\meter} \end{bmatrix}^{\Transpose}\big\} $ for constituting a closed barrier $ \Sbarr $.}
        \label{fig:025}
    \end{subfigure}
    \hfill
    \begin{subfigure}[t]{0.329\textwidth}
        \captionsetup{margin=5pt}
        \centering
        \includegraphics{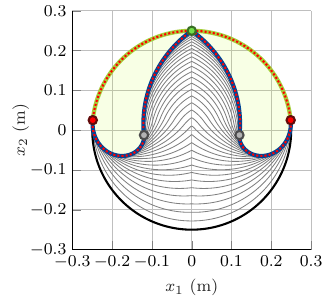}
        \caption{For $ \sv[\denoteLf] = \SI{0.10}{\meter \per \second} $, there is no significant difference between the \gls*{TEB} resulting from the presented method and the \gls*{TEB} resulting from \gls*{fastrack}. The boundaries of the \glspl*{TEB} coincide.}
        \label{fig:025_comp}
    \end{subfigure}
    \hfill
    \begin{subfigure}[t]{0.329\textwidth}
        \captionsetup{margin=5pt}
        \centering
        \includegraphics{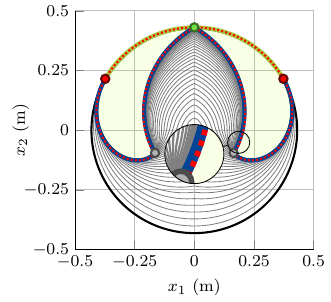}
        \caption{For $ \sv[\denoteLf] = \SI{0.50}{\meter \per \second} $, there is a difference between the boundary of the \gls*{TEB} resulting from the presented method and the \gls*{TEB} that is determined with \gls*{fastrack}. The boundaries do not fully coincide.}
        \label{fig:042_comp}
    \end{subfigure}
    \caption{The \gls*{TEB} resulting from the presented method is depicted by the green area. The black circle depicts the boundary of the captivity set $ \boundary\Scap $, the green line depicts the \gls*{IP}, and the blue line depicts the closed barrier $ \Sbarr $. The red dots depict the \gls*{BIP}, the green dots depict $ \overline\vcondSol(\sizeScap) $, and the gray dots depict where $ \su[\denoteHf]^{\denoteOpt} $ switches (see \eqref{eq:sols_up}). The gray lines illustrate how the numerical \gls*{fastrack} computations evolve, starting in $ \boundary\Scap $ and converging to the dotted red line, which depicts the boundary of the \gls*{TEB} determined with the method in \gls*{fastrack}.}
\end{figure*}
In this section, we demonstrate our method and compare it to the state of the art. For notational brevity, we neglect dependencies on time and, e.g., write $ \vx $ instead of $ \vx\oft $. All results obtained using the presented method are computed with Wolfram Mathematica to maintain numerical accuracy.
\subsection{Setting up the Relative System}\label{sec:num_ex_rel_sys}
Consider the dynamics of the homicidal chauffeur game studied in \autocite{patsko.2018}, \autocite{isaacs.1999}, \autocite{buzikov.2023}. The low-fidelity model used for planning employs the dynamics of the pedestrian, and the high-fidelity model used for tracking employs the dynamics of the chauffeur. The state vectors are given by $ \setlength{\arraycolsep}{2.5pt} \vx[\denoteLf] = \begin{bmatrix} \sxOne[\denoteLf] & \sxTwo[\denoteLf] \end{bmatrix}^{\Transpose} $ and $ \setlength{\arraycolsep}{2.5pt} \vx[\denoteHf] = \begin{bmatrix} \sxOne[\denoteHf] & \sxTwo[\denoteHf] & \so[\denoteHf] \end{bmatrix}^{\Transpose} $, and the dynamics are given by
\begin{equation}\label{eq:example_dynamics}
    \vxdot[\denoteLf] =
    \begin{bmatrix}
        \sv[\denoteLf] \sin\rBrack[{\su[\denoteLf]}]\\
        \sv[\denoteLf] \cos\rBrack[{\su[\denoteLf]}]\\
    \end{bmatrix}
    \text{ and }
    \vxdot[\denoteHf] = 
    \begin{bmatrix}
        \sv[\denoteHf] \sin\rBrack[{\so[\denoteHf]}]\\
        \sv[\denoteHf] \cos\rBrack[{\so[\denoteHf]}]\\
        \syaw[\denoteHf] \su[\denoteHf] 
    \end{bmatrix},
\end{equation}
with the inputs $ \Abs[{\su[\denoteHf]}] \leq 1 $ and $ \Abs[{\su[\denoteLf]}] \leq \pi $, the tracking model's maximum yaw rate $ 0 \leq \syaw[\denoteHf] < \infty $, and the constant velocities $ 0 \leq \sv[\denoteLf] < \sv[\denoteHf] $, which ensure that the tracking model is sufficiently fast to potentially remain within some distance from the planning model (see Remark~\ref{rem:powerful}). In $ \vmp[\denoteLf] $, there is only one parameter: $ \vmp[\denoteLf] = \sv[\denoteLf] $. As constraining the planning model per se in $ \sxOne[\denoteLf] $, $ \sxTwo[\denoteLf] $, or $ \su[\denoteLf] $ is impractical for motion planning, we use $ \pp = \vPB = \sv[\denoteLf] $.
The relative system's state vector results from \eqref{eq:rel_sys}:
\begin{equation}\label{eq:num_ex_rel_state}
    \vx = 
    \begin{bmatrix}
         \sxOne \\ \sxTwo \\ \so[\denoteHf]
    \end{bmatrix}
    =
    \relXTrans\rBrack[{\vx[\denoteLf],\vx[\denoteHf]}]
    \rBrack[
        \matchX
        \begin{bmatrix}
            { \sxOne[\denoteLf]} \\ {\sxTwo[\denoteLf]}
        \end{bmatrix} 
        -
        \begin{bmatrix}
            { \sxOne[\denoteHf]} \\ {\sxTwo[\denoteHf]} \\ {\so[\denoteHf]}
        \end{bmatrix} 
    ].
\end{equation}
We follow the convention in \autocite{patsko.2018}, \autocite{isaacs.1999}, \autocite{buzikov.2023} by using
\begin{equation}\label{eq:ne_relative_state}
    \setlength{\arraycolsep}{2.5pt}
    {\footnotesize
    \!\relXTrans\rBrack[{\vx[\denoteLf],\vx[\denoteHf]}] = 
    \begin{bmatrix*}[c]
        \cos\rBrack[{\so[\denoteHf]}] & -\sin\rBrack[{\so[\denoteHf]}] & \phantom{-}0\\
        \sin\rBrack[{\so[\denoteHf]}] &  \phantom{-}\cos\rBrack[{\so[\denoteHf]}] & \phantom{-}0\\
        0 & \phantom{-}0 & -1
    \end{bmatrix*}
    }
    \text{ and }
    \setlength{\arraycolsep}{2.5pt}
    {\footnotesize
    \matchX = \begin{bmatrix} 1 & 0 \\ 0 & 1\\ 0 & 0 \end{bmatrix}
    },
    \setlength{\arraycolsep}{5pt}
\end{equation}
which align the velocity $ \sv[\denoteHf] $ of the tracking model with the $ \sxTwo $-axis. The input $ \su[\denoteHf] $ is measured clockwise from the $ \sxTwo $-axis.
As the planning model fully controls the relative orientation $ \so[\denoteHf] - \su[\denoteLf] $, the state $ \so[\denoteHf] $ becomes obsolete. Thus, the relative state vector reduces to $ \vx = \begin{bmatrix} \sxOne & \sxTwo \end{bmatrix}^{\Transpose}\! $, yielding $ \vxdot = \reldynNE $:
\begin{equation}\label{eq:num_ex_rel_dyn_simp}
    \begin{bmatrix}
         \sxOnedot \\ \sxTwodot
    \end{bmatrix}
    =
    \begin{bmatrix*}[l]
        -           \sxTwo\syaw[\denoteHf] \su[\denoteHf] + \sv[\denoteLf]\sin\rBrack[{\su[\denoteLf]}] \\
        \phantom{-} \sxOne\syaw[\denoteHf] \su[\denoteHf] + \sv[\denoteLf]\cos\rBrack[{\su[\denoteLf]}] - \sv[\denoteHf] \\
    \end{bmatrix*}.
\end{equation}
\subsection{Solving the Related Game with Variable $ \pp $ and $ \sizeScap $}
To consider safety-critical constraints in the $ \sxOne $-$ \sxTwo $-plane, $ \vfselXcrit\rBrack[\vx] = \vx $ is required, yielding $ \vonScap = \vx $ for $ \vx \in \boundary\Scap $. Solving \eqref{eq:def_sp} yields the \gls*{IP} as $ \cBrack[{\vx\in\boundary\Scap \vertM \sizeScap \sv[\denoteLf]/\sv[\denoteHf] \leq \sxTwo}] $ and the \gls*{BIP}
\begin{equation}\label{eq:barr_init_cond}
    \SBIP =
    \cBrack[{\vx \in \boundary\Scap \vertM \sxOne = \sizeScap \sqrt{1-(\sv[\denoteLf]/\sv[\denoteHf])^2},\, \sxTwo = \sizeScap \sv[\denoteLf]/\sv[\denoteHf]}],
\end{equation}
containing two states symmetrical about the $ \sxTwo $-axis. Thus, there are two $ \vSolTrj[\pp,\sizeScap]^{\Ssemsur} $, each constituting one semipermeable surface $ \Ssemsur $. For determining the optimal $ \su[\denoteHf]^{\denoteOpt} = \sStrategyFB[\denoteHf]^{\denoteOpt}\rBrack[{\vx}] $ and $ \su[\denoteLf]^{\denoteOpt} = \sStrategyNA[\denoteLf]^{\denoteOpt}\rBrack[{\vx,\su[\denoteHf]}] $ that define $ \Ssemsur $, \eqref{eq:num_ex_rel_dyn_simp} is inserted into \eqref{eq:def_sem_sur}:
\begin{equation}\label{eq:ne_along_semsur}
    \begin{aligned}
        \min\limits_{\su[\denoteHf]\in\admissU[\denoteHf]}\max\limits_{\su[\denoteLf]\in\admissU[\denoteLf]}
        \big(\sonBarr[\denotesx] (-\sxTwo\syaw[\denoteHf] \su[\denoteHf] + \sv[\denoteLf]\sin\rBrack[{\su[\denoteLf]}])&\\
        +\sonBarr[\denotesy] ( \sxOne\syaw[\denoteHf] \su[\denoteHf] + \sv[\denoteLf]\cos\rBrack[{\su[\denoteLf]}] - \sv[\denoteHf] )\big)&
        = 0.
    \end{aligned}
\end{equation}
%
%
By exploiting the structure of \eqref{eq:ne_along_semsur}, the optimal $ \su[\denoteHf] $ follows from linear minimization, and the optimal $ \su[\denoteLf] $ from trigonometric maximization, requiring $ \vxdot[\denoteLf] $ to align with $ \vonBarr $ to maximize the inner product $ \sv[\denoteLf] (\sonBarr[\denotesx] \sin\rBrack[{\su[\denoteLf]}] + \sonBarr[\denotesy] \cos\rBrack[{\su[\denoteLf]}]) $ in \eqref{eq:ne_along_semsur}:
\begin{subequations}\label{eq:sols}
    \begin{align}
        \su[\denoteLf]^{\denoteOpt} &= \arcsin\rBrack[{\sonBarr[\denotesx] / \Euclidean[{\vonBarr}]}] = \arccos\rBrack[{\sonBarr[\denotesy] / \Euclidean[{\vonBarr}]}],\\
        \su[\denoteHf]^{\denoteOpt} &= \signFunc \rBrack[{\sonBarr[\denotesx] \sxTwo - \sonBarr[\denotesy] \sxOne}].\label{eq:sols_up}
    \end{align}
\end{subequations}
Using \eqref{eq:sols} and \eqref{eq:xi_diff} yields the adjoint dynamics
%
\begin{equation}\label{eq:ne_ev_onsemsur}
    \vonBarrdot = \begin{bmatrix} -\sonBarr[\denotesy] \syaw[\denoteHf] \su[\denoteHf]^{\denoteOpt} & \sonBarr[\denotesx] \syaw[\denoteHf] \su[\denoteHf]^{\denoteOpt} \end{bmatrix}^{\Transpose},
\end{equation}
corresponding to a planar rotation, which preserves the magnitude. Integrating \eqref{eq:ne_ev_onsemsur} yields the sinusoidal solution
\begin{equation}\label{eq:sol_onbarr}
    \vonBarr = 
    \begin{bmatrix}
        \sonBarr[\denotesx] \\ \sonBarr[\denotesy]
    \end{bmatrix}
    =
    \begin{bmatrix}
        \randomVar \cos(\syaw[\denoteHf]\su[\denoteHf]^{\denoteOpt} \cdot \rBrack[{\stime - \stimeBip}] + \randomAngle)\\
        \randomVar \sin(\syaw[\denoteHf]\su[\denoteHf]^{\denoteOpt} \cdot \rBrack[{\stime - \stimeBip}] + \randomAngle)\\
    \end{bmatrix},
\end{equation}
where $ \randomVar \in \Reals $ is the constant magnitude, $ \randomAngle \in \Reals $ is the angular offset, and $ \pm\syaw[\denoteHf] $ is the piecewise constant angular rate. At $ \stime = \stimeBip $, \eqref{eq:xi_cond} must hold. Equation \eqref{eq:barr_init_cond} and $ \vonScap = \vx $ yield
\begin{equation}\label{eq:num_ex_xi_cond}
    \randomAngle = \arcsin \!\left(\frac{\sizeScap \sv[\denoteLf]}{\randomVar\sv[\denoteHf]}\right) = \arccos (\sizeScap \sqrt{1-(\sv[\denoteLf]/\sv[\denoteHf])^2}/\randomVar).
\end{equation}
As only the direction of $ \vonBarr $ is relevant, $ \sizeScap/\randomVar $ is neglected in \eqref{eq:num_ex_xi_cond}. The trajectories $ \vSolTrj[\pp,\sizeScap]^{\Ssemsur} $ result from solving
\begin{equation}\label{eq:num_ex_sem_sur}
    \begin{bmatrix}
        \sxOnedot \\ \sxTwodot
    \end{bmatrix}
    =
    \begin{bmatrix*}[l]
        -          \sxTwo\syaw[\denoteHf] \su[\denoteHf]^{\denoteOpt} + \sv[\denoteLf]\cos(\syaw[\denoteHf]\su[\denoteHf]^{\denoteOpt} \cdot \rBrack[{\stime - \stimeBip}] + \randomAngle) \\
        \phantom{-} \sxOne\syaw[\denoteHf] \su[\denoteHf]^{\denoteOpt} + \sv[\denoteLf]\sin(\syaw[\denoteHf]\su[\denoteHf]^{\denoteOpt} \cdot \rBrack[{\stime - \stimeBip}] + \randomAngle) - \sv[\denoteHf] \\
    \end{bmatrix*}
\end{equation}
for both states in \eqref{eq:barr_init_cond} that must hold at $ \stime = \stimeBip $ (see \eqref{eq:rel_dyn_barr}, \eqref{eq:x_cond}). The two resulting $ \vSolTrj[\pp,\sizeScap]^{\Ssemsur} $ are symmetrical about the $ \sxTwo $-axis.

For better comparison with \gls*{fastrack}, we require $ \setlength{\arraycolsep}{2.5pt}\overline\vcondSol(\sizeScap) = \big\{\begin{bmatrix} 0 & \sxTwo[\denoteZero] \end{bmatrix}^{\Transpose}\big\},\, \Abs[{\sxTwo[\denoteZero]}] \leq \sizeScap $ so that the resulting \gls*{TEB} is connected. Optimizing $ \pp $ w.r.t. \eqref{eq:dderrmeas}, \eqref{eq:barr_inside_scap}, \eqref{eq:goes_int} yields $ \setlength{\arraycolsep}{2.5pt} \overline\vcondSol(\sizeScap) = \big\{\begin{bmatrix} 0 & \sizeScap \end{bmatrix}^{\Transpose}\big\} $. So
\begin{equation}\label{eq:ne_condition}
    \vSolTrj[\pp,\sizeScap]^{\Ssemsur} \big\vert_{\stime = \hat{\stime}} = \begin{bmatrix} 0 & \sizeScap \end{bmatrix}^{\Transpose}
\end{equation}
must hold for both $ \vSolTrj[\pp,\sizeScap]^{\Ssemsur} $.
\subsection{Addressing \ref{O.0}, \ref{O.1}, and \ref{O.2} Using the Game}\label{sec:ne_solving}
In the sequel, we use $ \sv[\denoteHf] = \SI{1}{\meter \per \second} $ and $ \syaw[\denoteHf] = \SI{2\pi}{\radian \per \second} $.
\subsubsection{Addressing \ref{O.1}}\label{sec:ne_objective1}
Consider the safety margin to be given as $ \ssafetyMargin = \SI{0.25}{\meter} $. Solving \eqref{eq:ne_condition} with $ \sizeScap = \ssafetyMargin $ results in $ \sv[\denoteLf] \approx \SI{0.10}{\meter \per \second} $.\footnote{\label{fn:accurate_solution_exists}An accurate but lengthy solution exists.} The resulting \gls*{TEB} $ \Steb $ is depicted in \autoref{fig:025}.
\subsubsection{Addressing \ref{O.2}}\label{sec:ne_objective2}
The safety controller results from \eqref{eq:sols_up}:
\begin{equation}
            \vSCtrlLaw\rBrack[\vx] =
        \left\{\setlength{\arraycolsep}{3.5pt}
        \begin{array}{ll}
             \phantom{-}1   & \text{if } \vx \in \boundary\Steb \text{ at } \raisebox{.5pt}{\textcircled{\raisebox{-.9pt} {1}}} \text{ or } \raisebox{.5pt}{\textcircled{\raisebox{-.9pt} {3}}},\\
            -1 & \text{if } \vx \in \boundary\Steb \text{ at } \raisebox{.5pt}{\textcircled{\raisebox{-.9pt} {2}}} \text{ or } \raisebox{.5pt}{\textcircled{\raisebox{-.9pt} {4}}},\\
            \text{any } \su[\denoteHf] \in [-1,1] & \text{if } \vx \in \mathrm{int}(\Steb) \cup \text{\gls*{IP}}\backslash \text{\gls*{BIP}}.
        \end{array}
        \right.
\end{equation}
\subsubsection{Addressing \ref{O.0}}\label{sec:ne_objective0}
Consider the inverse of the task solved in \autoref{sec:ne_objective1}, with $ \sv[\denoteLf] = \SI{0.10}{\meter \per \second} $ given. Solving \eqref{eq:ne_condition} with $ \pp = \sv[\denoteLf] $ results in $ \ssafetyMargin = \sizeScap \approx \SI{0.25}{\meter} $.\footref{fn:accurate_solution_exists}
\subsection{Comparison to \glsentryshort{fastrack}}\label{sec:ne_comparison}
For comparing our method to the state of the art, \ref{O.0} is addressed with the method in \gls*{fastrack} for both $ \sv[\denoteLf] = \SI{0.10}{\meter \per \second} $ and $ \sv[\denoteLf] = \SI{0.50}{\meter \per \second} $. To this end, we use the helperOC\footnote{\url{https://www.github.com/HJReachability/helperOC}} toolbox and discretize the state space with $ \Delta \sxOne = \Delta\sxTwo = \SI{1e-3}{\meter} $ and $ \Delta\so[\denoteHf] = \SI{8.73e-3}{\radian} $.\footnote{Using \eqref{eq:num_ex_rel_state} results in numerical oscillations that corrupt the solution. Therefore, $ \tilde{\vx} = \matchX\vx[\denoteLf] - \vx[\denoteHf] $ is used, for which the \gls*{fastrack} solution is stable.} The resulting \glspl*{TEB} for $ \sv[\denoteLf] = \SI{0.10}{\meter \per \second} $ and $ \sv[\denoteLf] = \SI{0.50}{\meter \per \second} $ are depicted in \autoref{fig:025_comp} and \autoref{fig:042_comp}, respectively.
\subsection{Discussion}
All presented results are computed on an Intel\textregistered\ Core\texttrademark\ i9-12900 processor with $ \SI{128}{\giga \byte} $ RAM. The CPU time for determining the \gls*{TEB} with \gls*{fastrack} at the stated accuracy is $ \SI{3.13e5}{\second} $ for $ \sv[\denoteLf] = \SI{0.10}{\meter \per \second} $ and $ \SI{1.15e6}{\second} $ for $ \sv[\denoteLf] = \SI{0.50}{\meter \per \second} $. The CPU time for addressing \ref{O.0} with our method is $ \SI{18.17}{\second} $ (see \autoref{sec:ne_objective0}), independently of the value of $ \sv[\denoteLf] $. The CPU time for addressing \ref{O.1} with our method is $ \SI{1.79}{\second} $ (see \autoref{sec:ne_objective1}), and the optimization for determining $ \overline{\vcondSol}(\sizeScap) $ takes less than $ \SI{1}{\second} $. The computation time of our method can be reduced to less than $ \SI{10}{\milli\second} $ by solving \eqref{eq:manifold} numerically using a C++ implementation, enabling online adaptation of the planning performance to a given safety margin to mitigate conservatism in safe motion generation.\footnote{Note that it can still be checked if the numerical accuracy of the result ensures safety by inserting the result into \eqref{eq:manifold}.}

For addressing \ref{O.0}, there is no significant difference between the \gls*{TEB} determined with the presented method and the \gls*{TEB} determined with \gls*{fastrack} for $ \sv[\denoteLf] = \SI{0.10}{\meter \per \second} $ (see \autoref{fig:025_comp}). However, for $ \sv[\denoteLf] = \SI{0.50}{\meter \per \second} $, there is a minor difference in the two \glspl*{TEB} (see the enlarged region in \autoref{fig:042_comp}), indicating the effects of the numerical methods involved in \gls*{fastrack}.\footnote{The discrepancy is hypothetically rooted in the kink in $ \Sbarr $ (at the states in which $ \su[\denoteHf]^{\denoteOpt} $ switches), which causes dissipation as described in \autocite[Sec. 2.2.1]{mitchell.2002}.} This demonstrates that, in the presented example, aside from numerical discrepancies, our method yields the same \gls*{TEB} as the method in \gls*{fastrack}, which is the least conservative method for addressing \ref{O.0} in the literature. For other benchmark systems we have analyzed, the resulting \glspl*{TEB} also coincide, suggesting the hypothesis that the presented method yields the same \gls*{TEB} as \gls*{fastrack} for combinations of systems complying with Definition~\ref{def:model_hf} and Definition~\ref{def:model_lf}; however, a formal proof is still under development.

%% file: 01_Tex/07_conclusion.tex
\section{Conclusion}\label{sec:08_conclusion}
\enlargethispage{-2\baselineskip}
In this paper, we have presented a method that addresses the conservatism, computational effort, and limited numerical accuracy of existing frameworks and methods that ensure safety in online model-based motion generation. In contrast to existing approaches that determine a safety margin for a given pair of models, we adopt a different perspective and directly adapt the performance of the model used for planning to a given safety margin, thereby enabling the mitigation of conservatism in existing safe motion generation frameworks. By leveraging a captivity-escape game, a novel zero-sum differential game formulated in this paper, our method requires significantly less computation time than the state of the art, and in contrast to the state-of-the-art method, our method yields numerically accurate results. With these benefits, our method complements established frameworks, making them even more effective.

%% file: 01_Tex/99_bib.tex
\section*{References}

\printbibliography[heading=none]